\documentclass[11pt,a4paper]{article}
\usepackage[latin1]{inputenc}
\usepackage[T1]{fontenc}
\usepackage[english]{babel}
\usepackage{a4}
\usepackage[numbers, sort&compress]{natbib}
\usepackage{amsbsy}   
\usepackage{amsfonts}   
\usepackage{amsmath}         
\usepackage{amssymb}    
\usepackage{amsthm}
\usepackage{authblk}
\usepackage{bm}
\usepackage{enumitem}
\usepackage{mathtools}
\usepackage{thmtools}
\usepackage[pdftex, breaklinks]{hyperref}
\allowdisplaybreaks
\newcommand{\xtl}{{X^n_{i-1}}}
\newcommand{\xtr}{X^n_i}
\newcommand{\xtlj}{{X^n_{j-1}}}
\newcommand{\xtrj}{X^n_j}
\newcommand{\thetan}{{\theta_0}}
\newcommand{\varepsilonn}{{\varepsilon_0}}
\newcommand{\alphan}{{\alpha_0}}
\newcommand{\betan}{{\beta_0}}
\newcommand{\thetak}{{\theta_k}}
\newcommand{\tminus}{{t_{i-1}^n}}
\newcommand{\tplus}{{t_i^n}}
\newcommand{\EE}{{\mathbb E}}
\newcommand{\PP}{{\mathbb P}}
\newcommand{\RR}{{\mathbb R}}
\newcommand{\NN}{{\mathbb N}}
\newcommand{\M}{{\mathbf M}}
\renewcommand{\S}{{\mathbf S}}
\newcommand{\W}{{\mathbf W}}
\newcommand{\X}{{\mathbf X}}
\newcommand{\aaa}{{\mathcal{A}}}
\newcommand{\cc}{{\mathcal{C}}} 
\newcommand{\dd}{{\mathcal{D}}}
\newcommand{\ff}{{\mathcal{F}}}  
\newcommand{\ii}{{\mathcal{I}}}  
\renewcommand{\ll}{{\mathcal{L}}} 
\newcommand{\nn}{{\mathcal{N}}}  
\newcommand{\pp}{{\mathcal{P}}}  
\newcommand{\qq}{{\mathcal{Q}}}  
\renewcommand{\ss}{{\mathcal{S}}} 
\newcommand{\xx}{{\mathcal{X}}} 
\newcommand{\ww}{{\mathcal{W}}} 
\declaretheorem[numberwithin = section]{theorem}
\declaretheorem[numberlike = theorem]{lemma}
\declaretheorem[numberlike = theorem]{corollary}
\declaretheorem[numberlike = theorem]{assumption}
\declaretheorem[numberlike = theorem]{condition}
\declaretheorem[numberlike = theorem, style = definition]{definition}
\declaretheorem[numberlike = theorem, style = remark]{example}
\declaretheorem[numberlike = theorem, style = remark]{remark}
\newcommand{\dqed}{{\leavevmode \unskip \penalty9999 \hbox{} \nobreak \hfill \quad \hbox{$\diamond$}}}
\newcommand{\cqed}{{\leavevmode \unskip \penalty9999 \hbox{} \nobreak \hfill \quad \hbox{$\circ$}}}
\numberwithin{equation}{section}
\numberwithin{theorem}{section}
\title{Estimating functions for jump-diffusions}
\author[1,2]{\small Nina Munkholt Jakobsen \thanks{nmja@dtu.dk}}
\author[2]{\small Michael S\o rensen \thanks{michael@math.ku.dk}}
\affil[1]{\footnotesize Department of Applied Mathematics and Computer
  Science, Technical University of Denmark, Richard Petersens Plads,
  Building 324, DK-2800 Kgs. Lyngby, Denmark}
\affil[2]{\footnotesize Department of Mathematical Sciences,
  University of Copenhagen, Universitetsparken 5, DK-2100 Copenhagen
  {\O}, Denmark}
\date{\today}
\begin{document} 
\maketitle
\begin{abstract}
Asymptotic theory for approximate martingale estimating functions is generalised to diffusions with finite-activity jumps, when the sampling frequency and terminal sampling time go to infinity. Rate optimality and efficiency are of particular concern. Under mild assumptions, it is shown that estimators of drift, diffusion, and jump parameters are consistent and asymptotically normal, as well as rate-optimal for the drift and jump parameters. Additional conditions are derived, which ensure rate-optimality for the diffusion parameter as well as efficiency for all parameters. The findings indicate a potentially fruitful direction for the further development of estimation for jump-diffusions.
{\bf Running title:} Estimation for jump-diffusions.
\end{abstract}
\newpage

\section{Introduction}
When modelling phenomena in continuous time, diffusions with jumps are a natural generalisation or improvement of continuous diffusion processes driven by Wiener noise, or of pure-jump processes. Jump-diffusion models find application in, among other fields, biology \cite{golden2017}, neuroscience \cite{musila1991, giraudo1997, patel2008, jahn2011, ditlevsen2013}, finance \cite{merton1976, dejong2001, kou2002}, and engineering \cite{hermann2018}. Statistical inference for diffusions with jumps raises a broad spectrum of intriguing challenges. The models have continuous-time dynamics, but sampling in continuous time is not feasible. Furthermore, a closed-form expression for the likelihood function based on discrete-time observations is not available, rendering maximum likelihood estimation impracticable. \medskip

A number of estimation approaches based on discrete-time observations have previously been proposed in the literature. A non-exhaustive list of references includes the following. In the context of parametric estimation, pseudo-likelihood methods involving, primarily, Gaus\-sian-in\-spi\-red approximations of the log-li\-ke\-li\-hood function have been considered \cite{shimizu2006, shimizu2006-3, ogihara2011, masuda2011, masuda2013, long2017}, as well as closed-form expansion of the transition densities \cite{yu2007, filipovic2013, li2016}, and approximations to maximum likelihood estimators obtained from the continuous-time likelihood function \cite{mai2014, gloter2018}. A quadratic variation-inspired estimation method was proposed in a semiparametric setting \cite{mancini2004}, and several non-parametric procedures \cite{bandi2003, shimizu2006-2, shimizu2008, shimizu2009, mancini2009, mancini2011, schmisser2014, wang2017, zhou2017} as well as a selection of simulation-based methods \cite{stramer2010, giesecke2017, guay2017, goncalves2017} have been considered. Finally, parametric estimation for diffusions with jumps based on observations made in continuous time has been investigated too \cite{soerensen1991}. \medskip

The present paper focuses on parametric estimation for an ergodic stochastic process $\X^\theta = (X^\theta_t)_{t\geq 0}$ with finite-activity jumps, using discrete-time observations. The process is assumed to take values in the open interval $\xx\subseteq \RR$, and solve a stochastic differential equation (SDE) of the form
\begin{align}
dX^\theta_t &= a(X^\theta_t, \theta)\,dt + b(X^\theta_t, \theta)\, dW_t + \int_\RR c(X^\theta_{t-}, z, \theta)\, N^\theta(dt,dz)\,,
\label{jumpSDE}
\end{align}
for $\theta$ in an open parameter set $\Theta\subseteq \RR^d$, $d\geq 1$. The drift, diffusion, and jump coefficients, denoted $a$, $b$, and $c$, respectively, are specified deterministic functions. As usual, $\X^\theta_- = (X^\theta_{t-})_{t\geq 0}$ is defined as the process of left limits of $\X^\theta$. The standard Wiener process $(W_t)_{t\geq 0}$ is assumed to be independent of $N^\theta(dt,dz)$, a time-homogeneous, finite-activity Poisson random measure on $[0,\infty) \times \RR$. Supposing that the stochastic process is observed at times $t_i^n = i\Delta_n$, $i=0,1,\ldots,n$, $\Delta_n > 0$, we consider a high-frequency asymptotic scenario with an infinite time horizon: $\Delta_n \to 0$ and $n\Delta_n \to \infty$ as $n\to \infty$. Assuming also the existence of a true, unknown parameter $\thetan \in \Theta$, we put $\X = \X^\thetan$ in the following, and introduce the notation $X^\theta_{n,i} = X^\theta_{t_i^n}$ and $X^n_i = X_{t_i^n}$. \medskip

In \cite{shimizu2006, ogihara2011}, two very similar contrast functions were proposed for estimation in models of the type (\ref{jumpSDE}) with finite-activity jumps. Both papers assume a parameter $\theta$ separated into two components, one present in only the drift and jump terms of the SDE, the other only entering into the diffusion term. The corresponding estimators are rate optimal and efficient. An essential aspect of these contrast functions is the technique for deciding whether or not a jump is likely to have occurred in an observation interval $[\tminus,\tplus]$. This method has become standard in financial econometrics \cite{mancini2001, lee2008}. Models similar to \eqref{jumpSDE}, but allowing also infinite-activity jumps, were treated in \cite{masuda2011, masuda2013}. In these models, the parameter separates into a drift component entering into only the drift term of the SDE, and a noise component figuring in both the diffusion and jump terms. For estimation, specific choices of Gaussian quasi-likelihood functions were used, which are known to work well for diffusions without jumps. In the presence of jumps, under an asymptotic scenario very similar to the one in the present paper, these Gaussian quasi-likelihood estimators were found to be neither rate optimal for the noise parameter nor efficient for any part of the parameter. \medskip

In this paper, we study approximate martingale estimating functions which may be written on the form
\begin{align}
G_n(\theta) &= \frac{1}{n\Delta_n}\sum_{i=1}^n g(\Delta_n,\xtr,\xtl,\theta)\,.
\label{GnDef}
\end{align}
The estimating function is defined by a deterministic function $g(t,y,x, \theta)$, which satisfies an approximate martingale property. This entails that for all $\theta \in \Theta$, the conditional expectation $\EE(g(\Delta_n,X^\theta_{n,i}, X^\theta_{n,i-1}, \theta)\mid X^\theta_{n,i-1})$ is of order $\Delta_n^{\kappa_0}$ for some constant $\kappa_0\geq 2$. Estimators are obtained as solutions to the equation $G_n(\theta)=0$, and referred to as $G_n$-estimators. For example, the Gaussian quasi-likelihood estimators considered in \cite{masuda2011, masuda2013} can, under regularity assumptions, be formulated in terms of approximate martingale estimating functions. \medskip

Approximate martingale estimating functions for continuous diffusions have already been quite thoroughly studied, see e.g.\ \cite{bibby1995, kessler1999, jacobsen2001, jacobsen2002, uchida2004, efficient, jakobsen2017}. In particular, the existing theory includes high-frequency asymptotics for the estimators under an infinite time horizon \cite{efficient} as well as infill asymptotics \cite{jakobsen2017}. Both of these papers present simple conditions on the estimating functions which ensure rate optimality and efficiency. It should also be noted that a large part of the estimators proposed in the literature for continuous diffusions can be treated within the framework of approximate martingale estimating functions \cite{efficient}. \medskip

Compared to continuous diffusions, statistical theory for diffusions with jumps is still establishing itself, and many open questions are yet to be solved. Approximate martingale estimating functions constitute a tractable, rather comprehensive framework for the study of estimation for continuous diffusions. One would, therefore, expect research on approximate martingale estimating functions to provide fruitful insight into estimation theory for jump-diffusions. In this paper, we explore high-frequency asymptotics for general approximate martingale estimating functions, an area which has not previously been studied in the context of jump-diffusions. \medskip

First we establish, under mild assumptions, a general theorem on the existence and uniqueness of consistent, asymptotically normal $G_n$-estimators of the parameter $\theta$ in the SDE model (\ref{jumpSDE}). We also provide a consistent estimator of the asymptotic variance. Next, we investigate the question of rate optimality and efficiency of the estimators of $\theta$. Usually, the optimal rate of convergence and efficient asymptotic variance would be identified using results from the theory of local asymptotic normality. However, local asymptotic normality and, for infill asymptotics, local asymptotic mixed normality are ongoing areas of research for stochastic processes with jumps \cite{kawai2013-2, kawai2013, clement2015, becheri2016, kohatsuhiga2014, kohatsuhiga2017}. No results for general jump-diffusions have been established so far. Nonetheless, the optimal rates of convergence and the Fisher information matrices in the current setup are pretty clear. It can rather safely be conjectured that the optimal rate of convergence is $\sqrt{n\Delta_n}$ for drift and jump components of the parameter and $\sqrt{n}$ for diffusion components, and that the efficient asymptotic variance is as proposed in Section \ref{sec:optimal_jump}. These conjectures are motivated not only by local asymptotic normality results which cover particular submodels of (\ref{jumpSDE}) \cite{kawai2013-2, becheri2016, kohatsuhiga2014, kohatsuhiga2017}, but also by other asymptotic results \cite{soerensen1991, gobet2002, shimizu2006}. \medskip

Considering two separate cases, we give conditions on $g$ under which the corresponding approximate martingale estimating function $G_n(\theta)$ yields rate optimal and efficient estimators. First, we assume the model (\ref{jumpSDE}) with no unknown parameter in the diffusion coefficient, so that there is only a drift-jump parameter to be estimated. Next, we present the case where (\ref{jumpSDE}) has a two-dimensional drift-jump parameter and a one-dimensional diffusion parameter. The restriction on the dimension of the parameter is due to the observation that when the diffusion coefficient depends on an unknown parameter, the complexity of the conditions obtained for rate optimality and efficiency increases substantially with the dimension of the parameter. \medskip

For our jump-diffusion models, in addition to the simple rate-optimality and efficiency conditions found by \cite{efficient} in the context of continuous diffusions, several new conditions appear. An intuition for these results can be obtained by the following considerations. In the limit $\Delta_n \to 0$ (asymptotically), a full sample path of $\X$ is observed. In this hypothetical situation, all jump times may be identified as times $t$ for which $X_t \neq X_{t-}$, with jump sizes equal to $X_t-X_{t-}$. Consider, for example, (\ref{jumpSDE}) with a two-dimensional drift-jump parameter $\alpha$, and a one-dimensional diffusion parameter $\beta$. For this model, an approximate martingale estimation function is defined by a function $g = (g_\alpha, g_\beta)$, where $g_\alpha$ represents two coordinate functions related to the estimation of $\alpha$, and $g_\beta$ one coordinate function associated with $\beta$. If $g$ satisfies our conditions for rate optimal estimation of the diffusion parameter, then $g_\beta(0,X_t,X_{t-}, \theta) = 0$ when $t$ is a jump time. In other words, asymptotically, jumps in the data are not used for the estimation of $\beta$. When applied to continuous parts of the trajectory of the jump diffusion,  $g_\beta$ takes on the form used to define a rate optimal approximate martingale estimating function for the diffusion parameter of the corresponding continuous diffusion. When $g$ also satisfies the conditions ensuring efficient estimation, then, asymptotically, $g_\alpha$ too distinguishes perfectly between jumps and continuous parts of the trajectory of the process. Specifically, for a jump time $t$, $g_\alpha(0,X_t,X_{t-}, \theta)$ has the form of the jump-related term in the score function of the continuously sampled jump-diffusion process; see \cite{soerensen1991}. At non-jump times, $g$ has the structure ensuring an efficient estimating function for the drift and diffusion parameters of the corresponding continuous diffusion. \medskip

In this paper, we extend the established framework of high-frequency asymptotics for approximate martingale estimating functions for continuous diffusions to include jump-diffusion models. In particular, we maintain mathematically appealing assumptions regarding smoothness of the estimating functions. Under these assumptions, it is straightforward to obtain consistent, asymptotically normal estimators of the parameters of the jump-diffusions. The conditions for rate optimality and efficiency, however, impose further (inadvertently strict) restrictions, in terms of which models allow these conditions to be satisfied simultaneously with the smoothness assumptions on the estimating functions. Nonetheless, the conditions are enlightening. Taking into consideration the proofs used to arrive at our results, it is hard to imagine asymptotically well-performing estimators for jump-diffusions, which do not, essentially, conform to the conditions outlined above. Our findings indicate a path for future research in a framework where the estimating function is not required to be an approximate martingale estimating function satisfying the usual smoothness conditions. It seems likely that in conjunction with appropriate jump filtering, the extensive class of rate optimal and efficient approximate martingale estimating functions for continuous diffusions might be used to achieve rate optimal and efficient estimation of drift and diffusion parameters of jump-diffusions too. \medskip

The structure of the paper is as follows: Section~\ref{sec:jump:prelim} presents definitions, notation, and terminology used throughout the paper, as well as the main assumptions imposed on the jump-diffusion model and estimating functions. Section~\ref{sec:general_jump} presents the general theorem on the existence and asymptotics of consistent estimators based on approximate martingale estimating functions. Section~\ref{sec:optimal_jump} is devoted to criteria for rate optimality and efficiency of estimators of drift-jump and diffusion parameters. Section~\ref{sec:jump_proofs} contains central lemmas used to prove the main theorems, the proofs of these theorems, and the proofs of the lemmas. Appendix~\ref{app_lemmas} consists of technical auxiliary results used in the proofs of the aforementioned lemmas.

\section{Preliminaries}\label{sec:jump:prelim}
In this section we introduce basic notation, definitions, and regularity assumptions. Transposition of a matrix $M$ is denoted by $M^\star$, and $\Vert M \Vert$ represents the Euclidean norm. We denote by $I_p$ the $p\times p$ identity matrix. For any $\RR^p$-valued function $f$, let $f=(f_1, \ldots, f_p)^\star $, where $f_j$ denotes the $j$th coordinate function. For an $\RR^q$-valued argument $u$, let $\partial_{u_k} f_j$ denote the partial derivative of $f_j$ with respect to $u_k$ and be the $jk$th element of the $p\times q$ matrix $\partial_u f$. Let $f^2 = (f_1^2, \ldots, f_d^2)^\star $. For a $p\times q$ matrix-valued function $F = (F_{jk})$, we define $\partial_u F = (\partial_u F_{jk})$ for real-valued $u$ and $F^2 = (F^2_{jk})$. \medskip

Let $\Delta_0 = \max\{\Delta_n: n\in \NN \}$. Generic, strictly positive, real-valued constants are denoted by $C$. These constants may have dependencies emphasised by subscripts, and may also depend, implicitly, on e.g.\ $\theta_0$, $\Delta_0$, and $d$, but never on the sample size $n$. Choose $\varepsilon_0 > 0$ and let $(0,\Delta_0)_{\varepsilon_0} = (0-\varepsilon_0, \Delta_0+\varepsilon_0)$. A function $f: (0,\Delta_0)_{\varepsilon_0} \times \xx^2 \times \Theta \to \RR$ is said to be of polynomial growth in $x$ and $y$, uniformly for $t\in (0,\Delta_0)_\varepsilonn$ and $\theta$ in compact, convex sets if the following holds: For each compact, convex set $K \subseteq \Theta$, there exist constants $C_K>0$ such that for all $x,y \in \xx$,
\begin{align*}
  \sup_{t \in (0,\Delta_0)_{\varepsilon_0},\,\theta \in K} \left|f(t,y,x, \theta)\right|
  \leq C_K\left(1 + |x|^{C_K} + |y|^{C_K}\right)\,,
\end{align*}
or, equivalently,
\begin{align*}
  \sup_{t \in (0,\Delta_0)_{\varepsilon_0},\,\theta \in K} \left|f(t,y,x, \theta)\right|
  \leq C_K\left(1 + |x|^{C_K}\right)\left(1+|y|^{C_K}\right) \,.
\end{align*}
We use $R(t,y,x, \theta)$ to denote generic functions defined on $(0,\Delta_0)_{\varepsilon_0} \times \xx^2 \times \Theta$, which have coordinate functions of polynomial growth in $x$ and $y$, uniformly for $t\in (0,\Delta_0)_{\varepsilon_0}$ and $\theta$ in compact, convex sets. In the same manner as $C$, $R$ may have both explicit and implicit dependencies. Functions $R(t,x, \theta)$, $R(y,x, \theta)$, and $R(t,x)$ are defined correspondingly.
\begin{definition}
We denote by $\cc_\text{pol}^\infty ( (0,\Delta_0)_\varepsilonn\times \xx^2 \times \Theta)$ the class of real-valued functions $f(t,y,x, \theta) \in \cc^\infty ((0,\Delta_0)_\varepsilonn\times \xx^2 \times \Theta)$ which satisfy that $f$ and its partial derivatives of all orders are of polynomial growth in $x$ and $y$, uniformly for $t \in (0,\Delta_0)_\varepsilonn$ and $\theta$ in compact, convex sets. The classes $\cc^\infty_\text{pol} ((0,\Delta_0)_\varepsilonn \times \xx \times  \Theta)$, $\cc^\infty_\text{pol}( \xx^2 \times \Theta )$, $\cc^\infty_\text{pol}( \xx\times \RR \times \Theta)$, $\cc^\infty_\text{pol}( \xx \times \Theta)$, and $\cc^\infty_\text{pol}( \xx )$ are defined analogously for functions of the form $f(t,x, \theta)$, $f(y,x,\theta)$, $f(y,\theta)$, and $f(y)$. \dqed
\end{definition}

\subsection{Model}
\label{subsection:model}
Consider the filtered probability space $(\Omega, \ff, (\ff_t)_{t\geq 0}, \PP)$, equipped with the $(\ff_t)_{t\geq 0}$-adapted standard Wiener process $\W = (W_t)_{t\geq 0}$ and the independent, time-ho\-mo\-ge\-ne\-ous Poisson random measure $N^\theta(dt,dz)$ on $[0,\infty) \times \RR$ with intensity measure $\mu_\theta(dt,dz) =\nu_\theta(dz)\,dt$. Here, $\nu_\theta$ is a L\'evy measure on $\RR$ with $\nu_\theta(\{0\})=0$ and $\nu_\theta (\RR) < \infty$. Let $U^\theta$ be an $\ff_0$-measurable random variable which is independent of $\W$ and $N^\theta$. The process $\X^\theta$ is assumed to solve (\ref{jumpSDE}) with the initial condition $X^\theta_0=U^\theta$. The drift, diffusion, and jump coefficients of the SDE, $a, b: \xx\times \Theta \to \RR$ and $c: \xx\times \RR\times \Theta \to \RR$, respectively, are assumed to be known, deterministic functions. We make the following assumptions, among other reasons, in order to ensure that $\X$ may be viewed as a c\`{a}dl\`{a}g, $(\ff_t)$-adapted Markov process.
\begin{assumption}
Suppose that
\begin{align*}
  a(y, \theta)\,, b(y,\theta)\in \cc^\infty_{\text{pol}}( \xx\times \Theta)\quad
  \text{ and }\quad
  c(y, z, \theta) \in \cc^\infty_\text{pol}( \xx\times \RR \times \Theta)\,.
\end{align*}
Furthermore, the following holds for all $\theta \in \Theta$.
\begin{enumerate}[label=(\roman{*}), ref=(\roman{*})]
\item For all $y\in \xx$, $b^2(y, \theta) > 0$.
\item There exist real-valued constants $C_\theta>0$ such that  for all $x,y \in \xx$ and $z\in \RR$,
\begin{align*}
  |a(x, \theta) - a(y, \theta)| + |b(x, \theta) - b(y, \theta)|
  &\leq C_\theta |x-y|
\end{align*}
and
\begin{align*}
  |c(x,z, \theta) - c(y,z, \theta)|
  &\leq C_\theta |x-y|(1+|z|^{C_\theta})\,.
\end{align*}
\item \label{boundedmoments_jump} For all $m\in \NN$,
\begin{align*}
\sup_{t\in [0,\infty)} \EE ( |X^\theta_t|^m) < \infty\,.
\end{align*}
\item \label{x6} $\X^\theta$ is ergodic. That is, there exists an invariant probability measure $\pi_\theta$, such that for any $\pi_\theta$-integrable function $f$,
\begin{align}
  \frac{1}{T} \int_0^T f(X^\theta_t)\, dt
  \overset{\pp}{\longrightarrow}
  \int_\xx f(x) \,\pi_\theta(dx)
\label{ergodictheorem}
\end{align}
 as $T\to \infty$. The measure $\pi_\theta$ has moments of all orders. \item The L\'evy measure $\nu_\theta$ has density $q(z,\theta) = \xi(\theta)p(z,\theta)$ with respect to a $\sigma$-finite measure $\tilde{\nu}$, where $p(z,\theta)$ is a probability density with
respect to $\tilde{\nu}$.
\end{enumerate}
Finally, for the density of the L\'evy measure:
\begin{enumerate}[label=(\roman{*}), ref=(\roman{*}), resume]
\item\label{x3}  It holds that $(\theta \mapsto q(z, \theta)) \in \cc^2(\Theta)$. For each compact, convex set $K\subseteq \Theta$, there exists $\varphi_K: \RR \to [0,\infty)$ measurable with
\begin{align*}
\int_\RR |z|^m \,\varphi_K(z)\, \tilde{\nu}(dz) < \infty
\end{align*}
for all $m\in \NN_0$, such that for all $z \in \RR$ and $\theta \in K$, 
\begin{align*}
  q(z,\theta) + \sum_{j=1}^2 \sum_{k=1}^d \vert \partial^j_\thetak q(z, \theta)\vert
  \leq \varphi_K(z)\,.
\end{align*}
\end{enumerate}\dqed
\label{assumptions_on_X_jump}
\end{assumption}
In the following, we put $\pi_\thetan = \pi$. Note that by Assumption~\ref{assumptions_on_X_jump}.\ref{x3}, $\nu_\theta$ has moments of all orders for all $\theta \in \Theta$. Assumption~\ref{assumptions_on_X_jump} is similar to assumptions of e.g.\ \cite{shimizu2006, ogihara2011, masuda2013}. See \cite{masuda2007, masuda2008} for conditions ensuring that an ergodic theorem of the form (\ref{ergodictheorem}) holds, and under which $\X$ has bounded moments as in Assumption~\ref{assumptions_on_X_jump}.\ref{boundedmoments_jump}.
\medskip

Assuming ergodicity ensures the following lemma, which can be proved in the same way as the non-uniform part of \cite[Lemma
8]{kessler1997}, using the Cauchy-Schwarz and Jensen's inequalities, and a version of \cite[Proposition 3.1]{shimizu2006} (see also \cite[p. 1598]{masuda2013}).
\begin{lemma}
Suppose that Assumption~\ref{assumptions_on_X_jump} holds, and that
for fixed $\theta \in \Theta$, the functions $x \mapsto f(x, \theta)$
and $x \mapsto \partial_x f(x,\theta)$ are continuous and of
polynomial growth in $x$ for $x\in\xx$. Then, pointwise for $\theta \in \Theta$,
\begin{align*}
\frac{1}{n} \sum_{i=1}^{n} f(\xtl, \theta) &\overset{\pp}{\longrightarrow}
\int_\xx f(x, \theta)\, \pi(dx)\,.
\end{align*}
\dqed
\label{ucp}
\end{lemma}
Now suppose that Assumption~\ref{assumptions_on_X_jump} holds, and let $\lambda \in \Theta$. Assume that $f(t,y,x,\theta)$, $f: ((0,\Delta_0)_{\varepsilon_0}\times \xx^2 \times \Theta) \to \RR$, and its partial derivatives $\partial^i_y f$, $i=1,2$, exist, are continuous, and are of polynomial growth in $x$ and $y$, uniformly for $t\in (0,\Delta_0)_{\varepsilon_0}$ and $\theta$ in compact, convex sets. Then, the infinitesimal generator $\ll_\lambda$ is defined by
\begin{align}
\begin{split}\label{infinitesimalgenerator2} 
&\hspace{-5mm} \ll_\lambda f(t,y,x, \theta) \\
&= a(y, \lambda)\partial_y f(t,y,x, \theta)+ \tfrac{1}{2} b^2(y,
\lambda)\partial^2_y f(t,y,x, \theta) \\
&\hspace{5mm} + \int_\RR (f(t,y+ c(y,z, \lambda),x, \theta) - f(t,y,x,
\theta))\, \nu_\lambda(dz)\,.
\end{split}
\end{align}
Often, the notation $\ll_\lambda f(t,y,x,\theta)=\ll_\lambda(f(t,\theta))(y,x)$ is used, and we put $\ll_\thetan = \ll$. Since $\nu_\lambda(\RR)<\infty$, there exist constants $C_{\lambda,\theta} > 0$ such that for all $t \in (0,\Delta_0)_{\varepsilon_0}$, $x,y \in \xx$, and $\theta \in \Theta$, 
\begin{align*}
  \int_\RR |f(t,y + c(y,z, \lambda),x, \theta) - f(t,y,x, \theta))|\, \nu_\lambda(dz)
  &\leq C_{\lambda,\theta}\left(1+|x|^{C_{\lambda,\theta}}+ |y|^{C_{\lambda,\theta}}\right)\,,
\end{align*}
ensuring, in particular, that (\ref{infinitesimalgenerator2}) is well-defined. More generally, it may be verified that integrals of the form $\int f(\cdot,z)\nu_\lambda(dz)$ inherit polynomial growth properties of $f$. The operator $\ll_\lambda$, always acting on the variable $y$, is defined correspondingly for e.g.\ functions of the form $f(y)$, $f(y,x,\theta)$, and $f(t,y,x,\mathbf{z}_k, \theta)$ where $\mathbf{z}_k\in \RR^k$.
In the latter case, the notation $\ll_\lambda f(t,y,x,\mathbf{z}_k, \theta)  = \ll_\lambda(f(t,\mathbf{z}_k, \theta) )(y,x)$ is used. \medskip

We let $\ll_\lambda^k f = \ll_\lambda(\ll_\lambda^{k-1} f)$ for $k\in \NN$ with $\ll_\lambda^0 f = f$. When $f$ is $\RR^d$-valued, and the generator is well-defined for each coordinate function, then $\ll_\lambda f = (\ll_\lambda f_1,\ldots, \ll_\lambda f_d)^\star $. If $F$ is a matrix-valued function, $\ll_\lambda F$ denotes the matrix with $ij$th element $\ll_\lambda F_{ij}$. \medskip

The infinitesimal generator notation is useful for expressing the
following lemma.
\begin{lemma}
Let Assumption~\ref{assumptions_on_X_jump} hold. For some $k\in \NN_0$, suppose that $f(y,x,\theta)$, $f: \xx^2\times \Theta \to \RR$, and its partial derivatives $\partial^i_y f$, $i=1,\ldots,2(k+1)$, exist, are continuous, and are of polynomial growth in $x$ and $y$, uniformly for $\theta$ in compact, convex sets. Then, for $0\leq t < t+ \Delta$, $\Delta\leq \Delta_0$,  and $\lambda \in \Theta$, 
\[ 
\EE\left ( f(X^\lambda_{t+\Delta}, X^\lambda_t, \theta) \mid X_t^\lambda\right) = \sum_{i=0}^k \frac{\Delta^i}{i!} \ll_\lambda^i f(X_t^\lambda, X_t^\lambda, \theta) + \Delta^{k+1} R_\lambda(\Delta,X^\lambda_t, \theta)\,, 
\]
where
\begin{eqnarray*}
\lefteqn{\Delta^{k+1} R_\lambda(\Delta,X^\lambda_t, \theta) =} \\ && \\ &&\int_0^{\Delta_0} \int_0^{u_1} \cdots \int_0^{u_k} \EE\left( \ll_\lambda ^{k+1} f(X^\lambda_{t+u_{k+1}},X_t^\lambda, \theta)\mid X^\lambda_t \right)\, du_{k+1}\cdots du_2\,du_1.
\end{eqnarray*}
\dqed
\label{lemma:expansion_jump}
\end{lemma} 
Lemma~\ref{lemma:expansion_jump} is effectively a jump-diffusion extension of the expression given by e.g.\ \cite[Lemma 1]{florenszmirou1989} for continuous diffusions. Notationally, the proof of Lemma~\ref{lemma:expansion_jump} is very similar to the proof of the continuous version, see \cite[Lemma 1.10]{soerensen2012} and \cite[Lemma 1]{kessler1997}, but it uses It\^{o}'s formula for diffusions with jumps and the infinitesimal generator (\ref{infinitesimalgenerator2}). The lemma is useful for verifying the approximate martingale property (\ref{jump_AMG}), and for creating approximate martingale estimating functions. It is also key to proving Lemma~\ref{lemma:conseq}, which emphasises two important properties of the estimating functions considered here.

\subsection{Estimating Functions}
\begin{definition}
Let $g: (0,\Delta_0)_{\varepsilon_0} \times \xx^2 \times \Theta \to \RR^d$. Suppose there exists a constant ${\kappa_0} \geq 2$, such that for all $n\in \NN$, $i=1,\ldots,n$, and $\theta \in \Theta$, 
\begin{align}
\EE\left( g(\Delta_n,X^\theta_{n,i},X^\theta_{n,i-1}, \theta) \mid   X^\theta_{n,i-1}\right) &= \Delta_n^{\kappa_0} R_\theta(\Delta_n, X^\theta_{n,i-1})\,.
\label{jump_AMG}
\end{align}
Then, (\ref{GnDef}) defines an approximate martingale estimating function. \dqed
\label{def:AMEF_jump} 
\end{definition}
A $G_n$-estimator $\hat{\theta}_n$ is obtained as a solution to the estimating equation $G_n(\theta) = 0$, see also \cite[Definition 2.3]{jakobsen2017}. When (\ref{jump_AMG}) is satisfied with $R_\theta(t,x)$ identically equal to zero, (\ref{GnDef}) is referred to as a martingale estimating function. \medskip

Let $M_n$ be any invertible $d\times d$ matrix with real entries which may depend on, e.g., $\Delta_n$ and $\thetan$. Then, $G_n(\theta)$ and $M_nG_n(\theta)$ produce identical estimators of $\theta$. These estimating functions are considered versions of each other. It is enough that one version satisfies the assumptions set forth in this paper. \medskip

We make the following assumptions about the function $g$, which defines our estimating function (\ref{GnDef}).
\begin{assumption}
Choose some $\varepsilon_0 > 0$.
\begin{enumerate}[label=(\roman{*}),ref=(\roman{*})] 
\item The function $g(t,y,x, \theta)$ satisfies Definition~\ref{def:AMEF_jump} for some ${\kappa_0} \geq 2$. 
\item \label{g_poly} For $j=1,\ldots,d$, it holds that
\begin{align*}
g_j(t,y,x, \theta) \in \cc^\infty_\text{pol}((0,\Delta_0)_{\varepsilon_0}\times \xx^2\times\Theta)\,.
\end{align*}
\item  For $k = 0,1$, and all $t \in (0,\Delta_0)_{\varepsilon_0}$, $x,y \in \xx$, and $\theta \in \Theta$, the expansion
\begin{align}
&\hspace{-5mm} \partial_\theta^kg(t, y, x, \theta) \nonumber\\
&= \partial_\theta^kg(0,y,x, \theta) + t \partial_\theta^kg^{(1)}(y,x, \theta) + \tfrac{1}{2}t^2\partial_\theta^kg^{(2)}(y,x, \theta) + t^3 R(t,y ,x, \theta)
\label{g2}
\end{align}
holds, where $g_j^{(i)}(y,x, \theta) = \partial^i_t g_j(0,y,x,\theta)$. \dqed
\end{enumerate}
\label{assumptions_on_g_jump}
\end{assumption}
In this paper, the assumptions of $\cc^\infty$-smoothness and polynomial growth, together with assumptions on the moments of, e.g., $X^\theta_t$, $\nu_\theta$, and $\pi_\theta$, serve to simplify the exposition and proofs, and could be relaxed; see also \cite[Chapter 3]{phdthesis}.
\medskip

For use in the proofs, we note the following lemma.
\begin{lemma}
Suppose that Assumptions~\ref{assumptions_on_X_jump} and \ref{assumptions_on_g_jump} hold. Then, for all $x\in \xx$ and $\theta \in \Theta$, 
\begin{align*}
g(0,x,x, \theta) = 0 \quad\text{ and }\quad g^{(1)}(x,x, \theta) = -\ll_\theta\left( g(0,\theta)\right)(x,x)\,.
\end{align*}\dqed
\label{lemma:conseq}
\end{lemma}
Lemma
\ref{lemma:conseq} is similar to \cite[Lemma 2.3]{efficient}, to which we refer for a proof.

\section{General asymptotic theory}\label{sec:general_jump}
This section presents the general theorem on the existence and properties of consistent estimators of $\theta$ in the model (\ref{jumpSDE}), based on approximate martingale estimating functions.
\begin{assumption}
Let 
\begin{align*}
A(\lambda, \theta) &= \int_\xx \big( \ll_\theta(g(0, \lambda)) (x,x) -\ll_\lambda(g(0, \lambda)) (x,x)\big)\, \pi_\theta(dx)\\
B(\theta, \theta) &= \int_\xx\big(  \ll_\theta(\partial_\theta g(0, \theta))(x,x) - \partial_\theta \ll_\theta(g(0,\theta))(x,x) \big)\, \pi_\theta(dx)\\
C(\theta, \theta) &= \int_\xx \ll_\theta(gg^\star (0,\theta))(x,x)\,\pi_\theta(dx)\,.
\end{align*}
The following holds for all $\theta \in \Theta$.
\begin{enumerate}[label=(\roman{*}), ref=(\roman{*})]
\item \label{LdiffNZ} The $\RR^d$-vector $A(\lambda, \theta)$ is non-zero whenever $\lambda \neq \theta$. 
\item \label{dthetalimitNS} The $d\times d$ matrix $B(\theta, \theta)$ is non-singular.
\item \label{Cposdef} The symmetric $d\times d$ matrix $C(\theta,\theta)$ is positive definite.
\end{enumerate}\dqed
\label{assumptions_for_estimators}
\end{assumption}
\begin{theorem}
Consider the model given by (\ref{jumpSDE}). Suppose that Assumptions~\ref{assumptions_on_X_jump}, \ref{assumptions_on_g_jump}, and \ref{assumptions_for_estimators} hold. If (\ref{jump_AMG}) holds with $R_\theta(t,x)$ not identically equal to $0$, suppose also that $n\Delta_n^{2{\kappa_0}-1} \to 0$ as $n\to \infty$. Then,
\begin{enumerate}[label=(\roman{*}), ref=(\roman{*})]
\item \label{main2} there exists a consistent $G_n$-estimator $\hat{\theta}_n$. Choose any compact, convex set $K\subseteq \Theta$ with $\theta_0 \in \text{int}\, K$, where $\text{int}\, K$ denotes the interior of $K$. Then, $\hat{\theta}_n$ is eventually unique in $K$, in the sense that for any $G_n$-estimator $\tilde{\theta}_n$ with $\PP(\tilde{\theta}_n \in K) \to 1$ as $n\to \infty$, it holds that $\PP(\hat{\theta}_n\neq\tilde{\theta}_n) \to 0$ as $n\to \infty$.
\item \label{main3} for any consistent $G_n$-estimator $\hat{\theta}_n$, it holds that
\begin{align*}
\sqrt{n\Delta_n}( \hat{\theta}_n - \theta_0) \overset{\dd}{\longrightarrow} \nn_d(0,V(\theta_0))\,,
\end{align*}
where $V(\theta_0) = B(\thetan, \thetan)^{-1}C(\thetan,\thetan)(B(\thetan,\thetan)^\star )^{-1}$ is positive definite.
\item \label{main4} for any consistent $G_n$-estimator $\hat{\theta}_n$,
\begin{align*}
\begin{split}
\widehat{V}_n &= n\Delta_n\left( \sum_{i=1}^n \partial_\theta g(\Delta_n, \xtr,\xtl, \hat{\theta}_n)\right)^{-1}\left(\sum_{i=1}^n gg^\star (\Delta_n, \xtr,\xtl,\hat{\theta}_n)\right)\\
&\hspace{5mm} \times \left( \sum_{i=1}^n \partial_\theta g^\star  (\Delta_n, \xtr,\xtl,\hat{\theta}_n)\right)^{-1}
\end{split}
\end{align*}
is a consistent estimator of $V(\theta_0)$, so
\begin{align*}
\sqrt{n\Delta_n} \,\widehat{V}_n^{-1/2} ( \hat{\theta}_n - \theta_0) \overset{\dd}{\longrightarrow} \nn_d(0,\text{I}_d)\,,
\end{align*}
where $\widehat{V}_n^{1/2}$ is the unique, positive semidefinite square root of $\widehat{V}_n$.
\end{enumerate}\dqed
\label{maintheo}
\end{theorem}
The proof of Theorem~\ref{maintheo} is given in Section \ref{section_proofofmaintheorem_jump}. 
The condition $n\Delta_n^{2{\kappa_0}-1} \to 0$ is necessary for the asymptotic normality in \ref{main3} and \ref{main4}, when the estimating function is not an exact martingale. These results follow from the asymptotic normality of the normalized estimating function, see Lemma \ref{uconvP}, which is proved using a central limit theorem for martingales. The normalized estimating function differs from a martingale by a term of order $(n\Delta_n^{2\kappa_0-1})^{1/2}$, which must necessarily vanish. In case $n\Delta_n^{2\kappa_0-1}$ converges to a constant different from zero, the asymptotic normal distribution has a mean different from zero.\medskip
\begin{example}
Consider the SDE
\begin{align*}
dX^\theta_t &= \tilde{a}(X^\theta_t, \alpha)\, dt + \tilde{b}(X^\theta_t, \beta)\sigma\, dW_t + \int_\RR \tilde{b}(X^\theta_{t-}, \beta)z\, (N-\mu)(dt,dz)\,.
\end{align*}
The drift parameter $\alpha$, and the diffusion-jump parameter $\beta$ are the unknown parameters to be estimated. Note that the Poisson random measure $N$ and its intensity measure $\mu$ do not depend on the parameters. For simplicity, let $\alpha \in A \subseteq \RR$ and $\beta \in B \subseteq \RR$ so that $d=2$ (the results generalise to larger $d$ as well). Put $\theta = (\alpha,\beta)^\star $ and $\Theta = A \times B$, and suppose that Assumption~\ref{assumptions_on_X_jump} holds. Furthermore, suppose that $\sigma^2 + \gamma_2 = 1$, where $\gamma_k$ denotes the $k$th moment of the L\'{e}vy measure $\nu$. \medskip

By Lemma~\ref{lemma:expansion_jump}, for $\theta \in \Theta$ and $0\leq t< t+\Delta$ with $\Delta \leq \Delta_0$,
\begin{align*}
\EE(X^\theta_{t+\Delta}\mid X^\theta_t) &= X^\theta_t + \Delta\,\tilde{a}(X^\theta_t, \alpha) +\Delta^2 R_\theta(\Delta, X^\theta_t)\\
\EE( (X^\theta_{t+\Delta}-X^\theta_t)^2\mid X^\theta_t) &= \Delta\,\tilde{b}^2(X^\theta_t,\beta) + \Delta^2 R_\theta(\Delta,X^\theta_t)
\end{align*}
so, under weak conditions on the functions $m_1(x, \theta)$ and $m_2(x,\theta)$, 
\begin{align*}
g(t,y,x, \theta) &= \begin{pmatrix} \displaystyle m_1(x,\theta)\left(y-x-t\tilde{a}(x, \alpha)\right) \\ \displaystyle m_2(x, \theta)\left(\left(y-x-t\tilde{a}(x, \alpha)\right)^2 - t\tilde{b}^2(x, \beta)\right)
\end{pmatrix}
\end{align*}
satisfies Assumption~\ref{assumptions_on_g_jump} with ${\kappa_0} = 2$. \medskip

Suppose also that Assumption~\ref{assumptions_for_estimators} holds, and that $n\Delta_n^3\to 0$ as $n\to \infty$. Then, by Theorem~\ref{maintheo}.\ref{main3}, for any consistent $G_n$-estimator $\hat{\theta}_n$,
\begin{align}
\sqrt{n\Delta_n} (\hat{\theta}_n-\thetan) &\overset{\dd}{\longrightarrow} \nn_2(0,V(\thetan))\,,
\label{example_limit}
\end{align}
where $V(\thetan) = B(\thetan, \thetan)^{-1}C(\thetan,\thetan)(B(\thetan,\thetan)^\star )^{-1}$ with
\begin{align*}
B(\thetan, \thetan) &= -\int_\xx \begin{pmatrix} m_1(x, \thetan)\partial_\alpha \tilde{a}(x, \alphan) & 0 \\ 0 & m_2(x, \thetan) \partial_\beta \tilde{b}^2(x, \betan) \end{pmatrix}\, \pi(dx)
\end{align*}
and
\begin{align*}
C(\thetan, \thetan) &= \int_\xx \begin{pmatrix} m_1^2(x, \thetan)\tilde{b}^2(x, \betan) & m_1 m_2(x, \thetan)\tilde{b}^3(x, \betan)\gamma_3 \\ m_2 m_1(x, \thetan)\tilde{b}^3(x, \betan)\gamma_3  & m_2^2(x, \thetan) \tilde{b}^4(x, \betan)\gamma_4 \end{pmatrix}\, \pi(dx)\,.
\end{align*}
\label{example:masuda2011}\cqed
\end{example}
The SDE and the estimating function used in Example~\ref{example:masuda2011} correspond to those considered in \cite{masuda2011}, incorporated into the current parametric framework. The result (\ref{example_limit}) is in accordance with \cite[Theorem 3.4]{masuda2011}. Similarly, in the particular case of 
quadratic approximate martingale estimating functions, the result in Theorem~\ref{maintheo}.\ref{main3} essentially follows from \cite[Theorem 2.9]{masuda2013} by interpreting the quasi-likelihood estimator proposed in \cite{masuda2013} as a quadratic approximate martingale estimating function.

\section{Rate optimality and efficiency}\label{sec:optimal_jump}
Here, we investigate conditions ensuring rate-optimal and efficient estimators. In Section~ \ref{condefs}, we discuss the optimal rate and Fisher information for jump-diffusion models. In Sections~\ref{generaldrift} and \ref{onedimdiff}, conditions are given, which are designed to ensure rate optimality and efficiency of $G_n$-estimators in two specific types of submodels of (\ref{jumpSDE}). The interpretation and implications of these conditions are discussed. \medskip

Suppose in the following that $A \subseteq \RR^{d_1}$ and $B\subseteq \RR^{d_2}$ with $d = d_1+d_2$, and put $\Theta = A\times B$. Consider, for $\alpha \in A$, $\beta \in B$, and $\theta = (\alpha,\beta)$, the SDE
\begin{align}
dX^\theta_t &= a(X^\theta_t, \alpha)\, dt + b(X^\theta_t, \beta)\, dW_t + \int_\RR c(X^\theta_{t-},z, \alpha) N^\alpha(dt, dz)\,,
\label{eqn:SDEoptimal2}
\end{align}
where $N^\alpha$ has intensity measure $\mu_\alpha(dt,dz) = \nu_\alpha(dz)\,dt$.
The parameters $\alpha$ and $\beta$ are referred to as the drift-jump and diffusion parameters, respectively. For simplicity, the following assumption is introduced.
\begin{assumption}
Let $c_{x, \alpha}(z) = c(x,z, \alpha)$. One of the following cases \ref{condanew} or \ref{condbnew} applies for all $x\in\xx$ and $\theta \in \Theta$.
\begin{enumerate}[label=(\alph{*}), ref=(\alph{*})]
\item\label{condanew} The dominating measure $\tilde{\nu}$ of the family of L\'evy   measures is Lebesgue measure. The set $\ww(x)  =  c_{x,\alpha}(\RR)$ is open and does not depend on $\alpha$. The mapping $z \mapsto c_{x, \alpha}(z)$ is bijective with a continuously differentiable inverse $w \mapsto c_{x, \alpha}^{-1}(w)$. Let
  \begin{align*} 
\varphi(x,w, \alpha) &= q(c_{x,\alpha}^{-1}(w), \alpha) |\partial_w c_{x,\alpha}^{-1}(w)|\,,\quad w \in \ww(x) \,,
\end{align*}
be the transformation of the L\'{e}vy density $q(\,\cdot\,, \alpha)$ by $z \mapsto  c_{x, \alpha}(z)$, and let $\eta_x$ denote Lebesgue measure on $\ww(x)$. 
\item\label{condbnew} The dominating measure $\tilde{\nu}$ is counting measure on an at most countable set $\qq \subset \RR$, and $c_{x,\alpha}(z) = c_x(z)$ for all $z\in \qq$. Define $\ww(x) = c_x(\qq)$ and
\begin{align*}
\varphi(x,w, \alpha) &= \sum_{z \in c_x^{-1}(\{w\})} q(z, \alpha)\,,
\end{align*}
and let $\eta_x$ denote counting measure on $\ww(x)$.
\end{enumerate}
In both cases, for all $x\in \xx$,
\[
\partial_\theta \hspace{-1mm}\int_{\ww(x)} \hspace{-5mm} g(0,x+w,x, \theta)\varphi(x,w,  \alpha)\,\eta_x(dw) = \int_{\ww(x)} \hspace{-5mm} \partial_\theta\big(                              g(0,x+w,x, \theta)\varphi(x,w, \alpha)\big)\, \eta_x(dw).
\]
\dqed
\label{assumptions_on_c} 
\end{assumption}

\subsection{Conjecture on rate optimality and efficiency}\label{condefs}
Local asymptotic normality has not yet been established for the general model (\ref{eqn:SDEoptimal2}) under the present observation scheme and asymptotic scenario. However, when $\X$ is ergodic, and under Assumption \ref{assumptions_on_c} and suitable regularity conditions, we conjecture the following: The jump-diffusion model is locally asymptotically normal with rate $\sqrt{n \Delta_n}$ for the drift-jump parameter $\alpha$, rate $\sqrt{n}$ for the diffusion parameter $\beta$, and Fisher information
\begin{align}
\ii(\theta) &= \begin{pmatrix} \ii_{1}(\theta) & 0 \\ 0 & \ii_{2}(\theta) \end{pmatrix}\,,
\label{eqn:Fisherinf} 
\end{align}
where
\begin{align}
\ii_{1}(\theta) &= \int_\xx\left( \frac{\partial_\alpha a(x, \alpha)^\star \partial_\alpha a(x,\alpha)}{b^2(x, \beta)} ) + \int_{\ww(x)}\hspace{-5mm} \frac{\partial_\alpha \varphi(x,w,\alpha)^\star \partial_\alpha \varphi(x,w, \alpha)}{\varphi(x,w,\alpha)}\,\eta_x(dw)\right)\, \pi_\theta(dx)
\label{I1}
\end{align}
and
\begin{align*}
  \mathcal{I}_{2}(\theta) &=  \int_\xx\frac{\partial_\beta b^2(x, \beta)^\star \partial_\beta b^2(x, \beta)}{2b^4(x, \beta)}\, \pi_\theta(dx)\,.
\end{align*}
The measure $\eta_x$ is either Lebesgue measure or counting measure on $\ww(x)$, see Assumption~\ref{assumptions_on_c}. In particular, we conjecture that a consistent estimator $\hat{\theta}^\star_n = (\hat{\alpha}^\star_n,\hat{\beta}^\star_n)$ of $\theta^\star = (\alpha^\star,\beta^\star)$ is rate optimal if
\begin{align}
\begin{pmatrix} \sqrt{n\Delta_n} (\hat{\alpha}_n -\alphan) \\
  \sqrt{n}(\hat{\beta}_n-\betan) \end{pmatrix}
&\overset{\dd}{\longrightarrow} Z \,,
\label{eqn:rateoptimality}
\end{align}
where $Z$ is a zero-mean, $d$-dimensional random vector with a positive definite covariance matrix, and that $\hat{\theta}_n$ is efficient if 
\begin{align}
Z \sim \nn_d(0,\ii(\thetan)^{-1} )\,.
\label{eqn:efficient}
\end{align}
The conjecture is motivated by the following observations. Here, the results quoted from the literature are given in a generality suitable for our framework and, to keep the discussion concise, regularity conditions are not included. Let $\widehat{\X}_T = (X_t)_{0 \leq t \leq T}$ denote continuous-time observations of the sample path of $\X$ over the interval $[0,T]$ for $T>0$, and let $\widehat{X}_n = (X_0^n, X_1^n, \ldots, X_n^n)$ denote discrete-time observations of $\X$ sampled as in this paper. \medskip

For continuous diffusions, local asymptotic normality for discrete-time observations $\widehat{X}_n$ with rate $\sqrt{n\Delta_n}$ for $\alpha$, $\sqrt{n}$ for $\beta$, and Fisher information (\ref{eqn:Fisherinf}) was shown in \cite[Theorem 4.1]{gobet2002}. There is no reason to expect it to be possible to estimate the diffusion parameter at a faster rate or more accurately in the jump-diffusion framework considered here. \medskip
 
In \cite{soerensen1991}, likelihood theory was developed for continuous-time data $\widehat{\X}_T$ for models of the type (\ref{eqn:SDEoptimal2}) with only a drift-jump parameter $\alpha$, under the assumption that the diffusion coefficient $b(x)$ is known. (This assumption is necessary for continuous time data.) The rate of convergence is $\sqrt{T}$ and, in case \ref{condanew} of Assumption \ref{assumptions_on_c}, it is seen from formulas (3.4), (3.6), and Corollary 3.3 of \cite{soerensen1991} that the Fisher information is given by (\ref{I1}). There is no reason to believe that the drift-jump parameter can be estimated at a faster rate or more accurately from discrete-time data. In \cite{gloter2018}, the local asymptotic normality property was established for continuous-time data when $c(x,z,\alpha) = \gamma(x) z$ and the Poisson random measure does not depend on $\alpha$. \medskip

For the general model (\ref{eqn:SDEoptimal2}), an estimator $\hat{\theta}_n$ based on $\widehat{X}_n$ was derived in \cite{shimizu2006}, which, in case \ref{condanew} of Assumption~\ref{assumptions_on_c}, satisfies (\ref{eqn:rateoptimality}) and (\ref{eqn:efficient}), provided that $n\Delta_n^2 \to \infty$. Thus, the conjectured rate and Fisher information can be achieved. Comparison to the result in \cite{gobet2002} indicates that the rate and Fisher information for the diffusion parameter must necessarily be as conjectured, while comparison to the result in \cite{soerensen1991} indicates that this is true of the drift-jump parameter too.\medskip

For certain submodels of (\ref{eqn:SDEoptimal2}), the following papers showed results on local asymptotic normality, which are in accordance with our conjecture. All results are for discrete-time data with infinite-horizon, high-frequency asymptotics. In \cite[Propositions 1 \& 2]{becheri2016}, the model (\ref{eqn:SDEoptimal2}) is investigated with $b(x, \beta) \equiv b(x)$, $c(x,z,\alpha) \equiv z$ in case \ref{condanew} of Assumption \ref{assumptions_on_c}. It is assumed that  $n\Delta_n^2 \to 0$ as $n\to \infty$. The model studied in \cite{becheri2016}\ is more general than our model in one respect: the jump intensity is allowed to be state-dependent. In \cite[Theorem 2.2]{kohatsuhiga2017}, the model (\ref{eqn:SDEoptimal2}) is studied with only a one-dimensional drift parameter $\alpha$ in either of the cases \ref{condanew} or \ref{condbnew} of Assumption~\ref{assumptions_on_c}. The diffusion coefficient and jump mechanism are assumed to be known. In \cite{kawai2013-2}, the Ornstein-Uhlenbeck process
\[
dX_t = -\alpha_1(X_t - \alpha_2)dt + \beta dW_t + \int_\RR z \, (N - \mu)(dt,dz)
\]
is considered, where the random measure $N$ and its intensity measure $\mu$ do not depend on $\alpha$. Both \ref{condanew} and \ref{condbnew} of Assumption~\ref{assumptions_on_c} are covered. Furthermore, local asymptotic normality is shown in \cite{tran2017}, when the random measure is given by a Poisson process with the intensity $\lambda$ a parameter to be estimated. Finally, the model 
\begin{align*}
dX_t &= (\alpha -\gamma)\, dt + \beta\, dW_t + \int_\RR z \, N(dt,dz)
\end{align*}
is investigated in \cite{kohatsuhiga2014} with $\nu_\gamma(dz) = \gamma \varepsilon_1(dz)$, where $\varepsilon_1$ is the degenerate probability measure with point mass in $1$, and the unknown parameter $\theta = (\alpha,\gamma,\beta)$ is three-dimensional.

\subsection{Drift-jump parameter}\label{generaldrift}
In this section, we consider the submodel of (\ref{eqn:SDEoptimal2}) given by
\begin{align}
dX^\alpha_t &= a(X^\alpha_t, \alpha)\, dt + b(X^\alpha_t)\, dW_t + \int_\RR
c(X^\alpha_{t-},z, \alpha) N^\alpha(dt, dz)\,, 
\label{eqn:SDE_diffusion}
\end{align}
where $\alpha \in A$, and $\Theta=A$ is a non-empty, open subset of $\RR^d$. Theorem~\ref{maintheo} already yields the conjectured optimal rate for estimators of the parameter $\alpha$. In order to ensure efficiency, we impose the following sufficient condition, which presupposes Assumption \ref{assumptions_on_c}.
\begin{condition}
For each $\alpha \in A$, there exists an invertible $d\times d$ matrix $K_\alpha$ such that for all $x\in\xx$, and $\eta_x$-almost all $w\in \ww(x)$,
\begin{align*}
  \partial_y g(0,x,x, \alpha) = K_\alpha \,\frac{\partial_\alpha a(x, \alpha)^\star }{b^2(x)}
  \quad \text{and}
  \quad g(0,x+w,x, \alpha) = K_\alpha\, \frac{\partial_\alpha \varphi(x,w, \alpha)^\star }{\varphi(x,w, \alpha)}\,.
\end{align*}
\dqed
\label{cond:eff_driftjump}
\end{condition}
Here, $\partial_y g(0,x,x, \alpha)$ denotes $\partial_y g(0,y,x, \alpha)$ evaluated in $y=x$.
Using Remark~\ref{re:BC_expr}, the following Corollary~\ref{cor:jump_eff_1} may easily be verified.
\begin{corollary}
Consider the model given by (\ref{eqn:SDE_diffusion}). Suppose that the assumptions of Theorem~\ref{maintheo}, as well as Assumption
\ref{assumptions_on_c} and Condition~\ref{cond:eff_driftjump} hold. Then,
any consistent $G_n$-estimator $\hat{\alpha}_n$ has asymptotic variance $\ii_1^{-1}(\alphan)$ given by (\ref{I1}).\dqed
\label{cor:jump_eff_1}
\end{corollary}
Thus $\hat{\alpha}_n$ is efficient, provided our conjecture is correct.
The first equation in Condition~\ref{cond:eff_driftjump} corresponds
to the condition given in \cite[Condition 1.2]{efficient} for
efficiency of estimators of the drift parameter of a continuous
diffusion. The second equation is of a type not seen in \cite{efficient}.
It is a jump-related condition on the
off-diagonal $y\neq x$ of $g(t,y,x, \theta)$ when $t=0$. In the limit $\Delta_n \to
0$, the full sample path of $\X$ is observed and $g(0,y,x, \alpha)$
and its derivatives are evaluated in $y=X_t$
and $x=X_{t-}$. For continuous diffusions, $X_t = X_{t-}$ for all $t$,
so conditions for $y\neq x$ are not relevant. For 
jump-diffusions, however, $X_t \neq X_{t-}$ whenever $t$ is a jump
time so, in this case, off-diagonal
conditions are expected.
\medskip

It is evident that an estimating function which satisfies Condition
\ref{cond:eff_driftjump} discriminates, asymptotically, between
situations where $X_t = X_{t-}$ and $X_t \neq X_{t-}$. When $X_t =
X_{t-}$, the function must behave like an efficient estimating function for the
drift parameter of the corresponding continuous diffusion. However, when $X_t \neq X_{t-}$, asymptotically, the function takes
on the form of the term associated with jumps in the score function of the
continuously sampled jump-diffusion process \cite{soerensen1991}.
In essence, the estimating function must, in the limit
$\Delta_n\to 0$, equal the  score function of the jumps at $(y,x)=(X_t,X_{t-})$ when $X_t \neq X_{t-}$. This  severely restricts the class of efficient estimators for
jump-diffusions in contrast to continuous diffusions.

\subsection{Drift-jump and diffusion parameter}\label{onedimdiff}
In this section, we consider the model (\ref{eqn:SDEoptimal2}) where $A \subseteq \RR^2$ and $B
\subseteq \RR$ are non-empty, open sets. Let  $g_\alpha = (g_1, g_2)^\star $ and $g_\beta = g_3$. For convenience, we introduce the following notation. Let $\ss_\alpha$ denote the support of the L\'evy-measure $\nu_\alpha$, and put $\tau_\alpha(y,z) = y + c(y,z,\alpha)$. Define the sets ${\cal M}_k(y,\alpha)$ iteratively by ${\cal M}_0(y,\alpha) = \{ y \}$, and ${\cal  M}_k(y,\alpha) = \tau_\alpha({\cal M}_{k-1} (y,\alpha), \ss_\alpha)$ for $k\in \NN$. The set ${\cal M}_k(y,\alpha)$ is the subset of $\xx$ that can be reached from the point $y$ by $k$ pure jumps, when $\alpha$ is the true drift-jump parameter.
The following Condition~\ref{assumption_final} is an additional condition for use in Theorem~\ref{maintheo_final}. The theorem improves on Theorem~\ref{maintheo}.\ref{main3}, in that it yields the conjectured optimal rate of convergence for consistent $G_n$-estimators $\hat{\theta}_n^\star = (\hat{\alpha}_n^\star,\hat{\beta}_n)$.
\begin{condition}
Suppose that for all $\tilde{\alpha} \in A$, $\theta \in \Theta$, and $x \in \xx$
\begin{align*} 
g_\beta(0,y,x, \theta) &= 0\,,\quad y \in {\cal M}_k(x,\tilde \alpha), \quad k=1,2,3,4\\
\partial_y g_\beta(0, y,x, \theta) &= 0\,,\quad y \in {\cal M}_k(x,\tilde \alpha), \quad k=0,1,2,3 \\
\partial^2_y \partial_\alpha g_\beta(0, y, x, \theta)&=0\,,\quad y \in {\cal M}_k(x,\tilde \alpha), \quad k=0,1\\
\partial_\alpha g_\beta^{(1)}(y,x, \theta) &=0, \quad y \in {\cal M}_1(x,\tilde \alpha).
\end{align*}
\dqed
\label{assumption_final}
\end{condition}
\begin{theorem}
Consider the model given by (\ref{eqn:SDEoptimal2}) with $A \subseteq \RR^2$ and $B
\subseteq \RR$. Suppose that all the assumptions of Theorem~\ref{maintheo} except Assumption~\ref{assumptions_for_estimators}.\ref{Cposdef} hold, and that Condition~\ref{assumption_final} holds. If (\ref{jump_AMG}) holds with $R_\theta(t,x)$ not identically equal to $0$, suppose also that $n\Delta_n^{2({\kappa_0}-1)} \to 0$ as $n\to \infty$. Let 
\begin{align*}
\begin{split}
B_1(\thetan, \thetan) &= - \int_\xx \partial_y g_\alpha(0,x,x, \thetan)\partial_\alpha a(x,\alphan) \,\pi(dx)\\
&\hspace{5mm} - \int_\xx \int_\RR\partial_y g_\alpha(0,x+c(x,z,\alphan),x,\thetan)\partial_\alpha c(x,z, \alphan) \,\nu_\alphan(dz)\,\pi(dx)\\
&\hspace{5mm} - \int_\xx \int_\RR g_\alpha(0,x+c(x,z, \alphan),x, \thetan)\partial_\alpha   q(z,\alphan)\,\tilde{\nu}(dz)\,\pi(dx)\,,
\end{split}\\[0.5em]
\begin{split}
B_2(\thetan, \thetan) &= -\int_\xx \tfrac{1}{2}\partial^2_y g_\beta(0,x,x, \thetan) \partial_\beta b^2(x, \betan) \, \pi(dx)\,,
\end{split}\\[0.5em]
\begin{split}
D_1(\theta, \theta)
&= \int_\xx b^2(x, \beta)\partial_y g_\alpha (\partial_y g_\alpha)^\star (0,x,x,\theta) \,\pi_\theta(dx)\\
&\hspace{5mm} + \int_\xx \int_\RR g_\alpha g_\alpha^\star (0,x+c(x,z, \alpha),x,\theta)\,\nu_\alpha(dz)\, \pi_\theta(dx)\,,
\end{split}\\[0.5em]
\begin{split}
D_2(\theta, \theta) 
&= \int_\xx \tfrac{1}{2}  b^4(x, \beta)  \left(\partial^2_y g_\beta(0,x,x, \theta)\right)^2 \,\pi_\theta(dx)\,,
\end{split}
\end{align*}
and assume that $D_1(\theta, \theta)$ is invertible, and $D_2(\theta, \theta)\neq 0$ for all $\theta \in \Theta$. Then, for any consistent $G_n$-estimator $\hat{\theta}_n$, it holds that
\begin{align}
\begin{pmatrix} \sqrt{n\Delta_n} (\hat{\alpha}_n - \alphan) \\  \sqrt{n}(\hat{\beta}_n-\betan) \end{pmatrix}
  &\overset{\dd}{\longrightarrow} \nn_3(0,V(\thetan))
\label{eqn:convD_final}
\end{align}
where
\begin{align*}
V(\thetan) = \begin{pmatrix}B_1(\thetan, \thetan)^{-1}D_1(\thetan, \thetan)(B_1(\thetan, \thetan)^\star )^{-1} & 0\\ 0 & B_2(\thetan, \thetan)^{-2}D_2(\thetan, \thetan)
\end{pmatrix}
\end{align*}
is positive definite. Furthermore, $\widehat{V}_n = \text{diag}( \widehat{V}_{n,1} , \widehat{V}_{n,2} )$ given by
\begin{align*}
\begin{split}
\widehat{V}_{n,1} &= n\Delta_n\left( \sum_{i=1}^n \partial_\alpha g_\alpha(\Delta_n,\xtr,\xtl,\hat{\theta}_n)\right)^{-1} \left( \sum_{i=1}^n g_\alpha g_\alpha^\star (\Delta_n,\xtr,\xtl, \hat{\theta}_n)\right)\\
&\hspace{5mm}\times\left( \sum_{i=1}^n \partial_\alpha g_\alpha^\star (\Delta_n,\xtr,\xtl,\hat{\theta}_n)\right)^{-1}
\end{split}\\[0.5em]
  \widehat{V}_{n,2} &= n\left( \sum_{i=1}^n \partial_\beta g_\beta(\Delta_n,\xtr,\xtl,\hat{\theta}_n)\right)^{-2} \sum_{i=1}^n g_\beta^2 (\Delta_n,\xtr,\xtl,\hat{\theta}_n)
\end{align*}
is a consistent estimator of $V(\thetan)$, so
\begin{align}
  \widehat{V}_n^{-1/2} \begin{pmatrix} \sqrt{n\Delta_n} (\hat{\alpha}_n - \alphan) \\  \sqrt{n}(\hat{\beta}_n-\betan) \end{pmatrix} &\overset{\dd}{\longrightarrow} \nn_3(0,\text{I}_3)\,,
                                                                                                                                      \label{eqn:convd_final2}
\end{align}
where $\widehat{V}_n^{1/2}$ is the unique, positive semidefinite square root of  $\widehat{V}_n$.\dqed
\label{maintheo_final}
\end{theorem}
The proof of Theorem~\ref{maintheo_final} is given in Section~ \ref{section_proofofmaintheorem_jump}. When $G_n(\theta)$ is not a martingale estimating function, the condition $n\Delta_n^{2({\kappa_0}-1)} \to 0$ is necessary for (\ref{eqn:convD_final}) and \eqref{eqn:convd_final2}. Here, $\Delta_n$ must converge to zero a bit faster than in Theorem \ref{maintheo}. The reason is that, under the conditions ensuring rate optimality, the third coordinate of the estimating function must be normalized differently from the first two coordinates. The third coordinate of the difference between the normalized estimating function and the normalized martingale is of order $( n\Delta_n^{2({\kappa_0}-1)})^{1/2}$, while the other two are of the same order as in Theorem \ref{maintheo}. \medskip

Making use of Remark~\ref{re:BC_expr}, it is evident that efficiency of the estimators according to our conjecture is ensured by the following additional Condition~\ref{cond:eff_final}, which presupposes Assumption~\ref{assumptions_on_c}.
\begin{condition}
For all $\theta \in \Theta$ there exists an invertible $2\times 2$ matrix $K_\theta^{(1)}$ and a constant $K_\theta^{(2)} \neq 0$, such that
\begin{align*}
&\partial_y g_\alpha (0,x,x,\theta) = K_\theta^{(1)} \frac{\partial_\alpha a(x,\alpha)^\star}{b^2(x,\beta)} \,,\quad \partial^2_y g_\beta(0,x,x,
                 \theta) = K_\theta^{(2)} \frac{\partial_\beta b^2 (x, \beta)}{b^4(x, \beta)}\,,\\
  \intertext{and}
&g_\alpha(0,x+w,x,\theta) = K_\theta^{(1)} \frac{\partial_\alpha \varphi(x,w,\alpha)^\star}{\varphi(x,w,\alpha)}
\end{align*}
for $\eta_x$-almost all $w\in \ww(x)$, and all $x\in\xx$.\dqed
\label{cond:eff_final}
\end{condition} 
\begin{corollary}
Suppose that the assumptions of Theorem~\ref{maintheo_final}, as well as Assumption
\ref{assumptions_on_c} and Condition~\ref{cond:eff_final} hold. Then, any consistent $G_n$-estimator $\hat{\theta}_n$ has asymptotic variance $\ii^{-1}(\thetan)$ given by (\ref{eqn:Fisherinf}). \dqed
\label{cor:final}
\end{corollary}
Thus $\hat{\theta}_n$ is efficient, provided our conjecture is correct.
The additional condition for efficiency in Condition \ref{cond:eff_final}, compared to Condition~\ref{cond:eff_driftjump}, is identical to the one identified in \cite{efficient} for the diffusion parameter of a continuous diffusion, and in \cite{jakobsen2017} for the diffusion parameter of a continuous diffusion in the case of infill asymptotics. \medskip

Condition~\ref{assumption_final} requires that the $g_\beta$ coordinate of $g$ as well as several of its derivatives vanish at a number of points depending on the jump dynamics of the process. This reveals that for many SDE models, rate optimal and efficient estimation of the diffusion parameter is not feasible when using the smooth estimating functions considered in this paper. In Theorem~\ref{maintheo_final}, the non-degeneracy condition on $D_2(\theta, \theta)$ requires that $\partial^2_y g_\beta(0,x,x, \theta)$ does not vanish $\pi_\theta$-almost surely for any $\theta$. This easily conflicts with Condition~\ref{assumption_final}. For example, for models where the first equation of the condition amounts to the requirement that $g_\beta(0,y,x, \beta)=0$ for all $x,y \in \xx$, it is clear that the non-degeneracy condition cannot be satisfied. \medskip

Finally, it may be noted that for special cases of (\ref{eqn:SDEoptimal2}) with only a one-dimensional diffusion parameter and no drift-jump parameter, Condition~\ref{assumption_final} may be reduced to its first and second equations involving $g_\beta$ and $\partial_y g_\beta$ for $k=1,2$ and $k=0,1$ respectively, see \cite[Section 3.4.4]{phdthesis} for further details. \medskip

For continuous diffusions, conditions under which an approximate martingale estimating function is rate optimal and efficient are quite straightforward, and it is easy to find estimating functions which satisfy the conditions. This was concluded in \cite{efficient} for the current sampling scheme, and in \cite{jakobsen2017} for infill asymptotics. The present paper demonstrates that the situation is more complex in the presence of jumps. In conclusion, the findings in this paper indicate that a way to obtain a larger number of rate optimal and efficient estimators might be to relax the smoothness conditions, and to allow the estimating function to discriminate more explicitly between intervals with jumps and intervals without jumps.  \medskip

Let us illustrate this conclusion by a brief consideration of the maximum likelihood estimator. In the particular case where $c(x,z,\alpha) = z$, the approximation to the transition density function proposed in \cite{yu2007} can be applied. If we assume that the drift function depends on $\alpha_1$ and the jump mechanism on $\alpha_2$, we find the approximation
\[
p(\Delta,y,x,\theta) = \frac{L(y,x,\alpha_1,\beta)}{\sqrt{2\pi\Delta} b(y,\beta)}(1 + O(\Delta))  +\varphi(y-x,\alpha_2) + O(\Delta^2),
\]
where
\begin{eqnarray*}
\lefteqn{L(y,x,\alpha_1,\beta)} \\ && = \exp \left(-\frac{(f_1(y,\beta)-f_1(x,\beta))^2}{2\Delta} +  f_2(y,\alpha_1,\beta) - f_2(x,\alpha_1,\beta) -\frac12 \log \frac{b(y,\beta)}{b(x,\beta)} \right)
\end{eqnarray*}
with $f_1(x,\beta) = \int^x 1/b (z,\beta) \,dz$ and $f_2(x, \alpha_1,\beta) 
= \int^x a(z,\alpha_1)/b^2(z,\beta) \,dz$. Note that in the case considered here, $\varphi(x,w,\alpha_2)$ does not depend on the first argument. Furthermore, $\Delta^{-1/2} \exp (-(f_1(y,\beta)-f_1(x,\beta))^2/(2\Delta) )$ functions as a kind of smooth indicator function which ensures that when increments are large, weight is mainly put on $\varphi$ when $\Delta$ is small.
\medskip
 
We consider the normalized score function of the form (\ref{GnDef}) with
\[
  g(\Delta,y,x,\theta) = \left( \partial_{\alpha_1} \log p(\Delta,y,x,\theta),\partial_{\alpha_2} \log p(\Delta,y,x,\theta), \Delta \partial_{\beta} \log p(\Delta,y,x,\theta)  \right)^\star.
\]    
Under weak regularity conditions, the score function is a martingale, see e.g.\ \cite{oebnms}. By straightforward, but rather long, calculations it follows from the approximation proposed in \cite{yu2007} that the derivatives of $g(\Delta,y,x,\theta)$ satisfy the equalities in Condition \ref{assumption_final} for all $x$ and $y$ and the efficiency conditions in Condition \ref{cond:eff_final} -- under suitable regularity assumptions on the remainder terms. However, the score function does not satisfy the smoothness condition in Assumption \ref{assumptions_on_g_jump}.\ref{g_poly}. Specifically, the function $y \mapsto g(0,y, x,\theta)$ is not continuous at $y=x$. Therefore, for instance, $\partial_y g_\alpha (0,x,x,\theta)$ in Condition \ref{cond:eff_final} must be interpreted as $\lim_{\Delta \rightarrow 0}  \partial_y g_\alpha (\Delta,x,x,\theta)$. This discontinuity explains why $g$ can satisfy both Condition \ref{assumption_final} and Condition \ref{cond:eff_final}. More importantly, the proofs of the present paper do not apply to the maximum likelihood estimator. Whereas our proofs and smoothness conditions are natural for the usual approximate martingale estimating functions based on conditional moments, it is an interesting theoretical question whether the proofs can be modified to cover less smooth estimating functions like the score function.

\section{Proofs}\label{sec:jump_proofs}
Section~\ref{sec:main_lemma_jump} presents lemmas used in Section~\ref{section_proofofmaintheorem_jump}, together with results from \cite{jacod2018}, to prove Theorems~\ref{maintheo} and \ref{maintheo_final}. Proofs of the lemmas are given in Section~\ref{sec:main_lemma_proofs}.

\subsection{Main Lemmas}\label{sec:main_lemma_jump}
\begin{lemma} 
Consider the model given by (\ref{jumpSDE}). Suppose that Assumptions~\ref{assumptions_on_X_jump} and \ref{assumptions_on_g_jump} hold. If (\ref{jump_AMG}) holds with $R_\theta(t,x)$ not identically equal to $0$, suppose, additionally, that $n\Delta_n^{2{\kappa_0}-1} \to 0$ as $n\to \infty$. For $\theta \in \Theta$, let 
\begin{align*}
A(\theta, \thetan) &= \int_\xx \left( \ll( g(0, \theta) )(x,x) - \ll_{\theta}( g(0, \theta) (x,x) \right)\, \pi(dx)\\
B(\theta, \thetan) &= \int_\xx\left(  \ll(\partial_\theta g(0, \theta))(x,x) - \partial_\theta   \ll_\theta(g(0,\theta))(x,x) \right)\, \pi(dx)\\
C(\theta, \thetan) &= \int_\xx \ll(gg^\star (0,\theta))(x,x)\, \pi(dx)\,.
\end{align*}
Then,
\begin{enumerate}[label=(\roman{*}),ref=(\roman{*})]
\item \label{ucon1} the mappings $\theta \mapsto A(\theta, \thetan)$, $\theta \mapsto B(\theta, \thetan)$, and $\theta \mapsto C(\theta, \thetan)$ are continuous on $\Theta$, with $A(\thetan, \thetan)=0$ and $\partial_\theta A(\theta, \thetan) = B(\theta, \thetan)$.
\item \label{ucon3} for all compact, convex sets $K\subseteq \Theta$,
\begin{align}
\sup_{\theta \in K} \left\Vert \frac{1}{n\Delta_n}\sum_{i=1}^n g(\Delta_n, \xtr, \xtl, \theta)   - A(\theta,\thetan)\right\Vert &\overset{\pp}{\longrightarrow} 0\,,\nonumber \\ 
\sup_{\theta \in K} \left\Vert \frac{1}{n\Delta_n}\sum_{i=1}^n \partial_\theta g (\Delta_n, \xtr, \xtl, \theta) -B(\theta, \thetan)\right\Vert&\overset{\pp}{\longrightarrow} 0\,, \label{convmainP2}\\
\sup_{\theta \in K} \left\Vert \frac{1}{n\Delta_n}\sum_{i=1}^n gg^\star(\Delta_n, \xtr, \xtl, \theta) - C(\theta,\thetan)\right\Vert &\overset{\pp}{\longrightarrow} 0\,. \nonumber
\end{align}
\item \label{ucon2} it holds that
\begin{align*}
  \frac{1}{\sqrt{n\Delta_n}} \sum_{i=1}^n g(\Delta_n,\xtr,\xtl, \thetan)
  &\overset{\dd}{\longrightarrow} \nn_d(0,C(\thetan,\thetan))\,.
\end{align*}
\end{enumerate}\dqed
\label{uconvP}
\end{lemma}
\begin{remark}
  Consider the model given by (\ref{eqn:SDEoptimal2}), and let $B(\thetan,\thetan)$ and $C(\thetan,\thetan)$ be as defined in Lemma~\ref{uconvP}. Under Assumptions~\ref{assumptions_on_X_jump}, \ref{assumptions_on_g_jump}, and \ref{assumptions_on_c}, we may write
\begin{align*}
\begin{split}
B(\thetan, \thetan) 
&= - \int_\xx \left( \partial_y g(0,x,x, \thetan)\partial_\theta  a(x, \alphan) + \tfrac{1}{2}\partial^2_y g(0,x,x, \thetan)\partial_\theta b^2(x,\betan) \right)\, \pi(dx) \\
&\hspace{5mm} - \int_\xx \int_{\ww(x)} g(0,x+w, x, \thetan)\partial_\theta \varphi(x,w, \alphan)\, \eta_x(dw) \, \pi(dx)\,,
\end{split}\intertext{and}
\begin{split}
C(\thetan,\thetan) 
&=  \int_\xx b^2(x, \betan) \partial_y g(\partial_yg)^\star (0,x,x, \thetan)\, \pi(dx) \\
&\hspace{5mm} +  \int_\xx \int_{\ww(x)} gg^\star  (0,x+w,x, \thetan) \varphi(x,w, \alphan)\, \eta_x(dw)\, \pi(dx)\,.
\end{split}
\end{align*}\cqed
\label{re:BC_expr}
\end{remark}
\begin{lemma}
  Consider the model given by (\ref{eqn:SDEoptimal2}), with $A\subseteq \RR^2$ and $B\subseteq \RR$. Suppose that Assumptions~\ref{assumptions_on_X_jump}, \ref{assumptions_on_g_jump}, and Condition~\ref{assumption_final} hold. If (\ref{jump_AMG}) holds with $R_\theta(t,x)$ not identically equal to $0$, we suppose, moreover, that $n\Delta_n^{2({\kappa_0}-1)} \to 0$ as $n\to \infty$. Let $\delta_n =   \textrm{diag}(\sqrt{n\Delta_n}, \sqrt{n\Delta_n}, \sqrt{n})$, and define the block diagonal matrix $D(\thetan, \thetan) = \text{diag}(D_1(\thetan, \thetan), D_2(\thetan, \thetan))$, where $D_1$ and $D_2$ are defined in Theorem \ref{maintheo_final}.
Then, 
\begin{enumerate}[label=(\roman{*}),ref=(\roman{*})] 
\item \label{con3_2} for all compact, convex sets $K\subseteq \Theta$,
\begin{align}
\sup_{\theta\in K} \left\Vert \frac{1}{n\Delta_n^{3/2}} \sum_{i=1}^n \partial_\alpha g_\beta(\Delta_n,\xtr,\xtl, \theta)\right\Vert&\overset{\pp}{\longrightarrow} 0\,,\label{ucon_newlab}\\
\sup_{\theta \in K}\left \vert \frac{1}{n\Delta_n^2} \sum_{i=1}^n g_\beta^2(\Delta_n,\xtr,\xtl, \theta) -D_2(\theta, \thetan) \right\vert &\overset{\pp}{\longrightarrow} 0\nonumber\,.
\end{align}
\item \label{con3_1} it holds that
\begin{align*}
\delta_n \frac{1}{n\Delta_n}\sum_{i=1}^n g(\Delta_n,\xtr,\xtl, \thetan)&\overset{\dd}{\longrightarrow} \nn_3\left(0, D(\thetan, \thetan)\right)\,.
\end{align*}
\end{enumerate}\dqed
\label{lemma:uconvP_final}
\end{lemma}

\subsection{Proofs of Main Theorems}
\label{section_proofofmaintheorem_jump}
\begin{proof}[\textbf{Proof of Theorem~\ref{maintheo}}]
  By Assumption~\ref{assumptions_for_estimators} and Lemma~\ref{uconvP}, $G_n(\theta)$ satisfies \cite[Condition $2.2$]{jacod2018} with $G(\theta) = A(\theta,\thetan)$ for all compact, convex sets $K\subseteq \Theta$ with $\theta_0 \in \text{int}\, K$, as well as \cite[Condition $2.10$]{jacod2018} with $A_n = \sqrt{n\Delta_n}I_d$, $Z = \nn_d(0,C(\theta_0, \thetan))$, and $W(\theta) = B(\theta,\thetan)$.
  \medskip

  By \cite[Theorem $2.5$]{jacod2018}, there exists a consistent $G_n$-estimator $\hat{\theta}_n$. 
  If $\tilde{\theta}_n$ is any $G_n$-estimator which satisfies that $\PP( \tilde{\theta}_n \in K) \to 1$ as $n\to \infty$, then $\tilde{\theta}_n$ is consistent by \cite[Theorem $2.7.(1)$]{jacod2018}, and  Theorem~\ref{maintheo}.\ref{main2} is a consequence of the uniqueness result in \cite[Theorem $2.5$]{jacod2018}. Theorem~\ref{maintheo}.\ref{main3} follows from \cite[Theorem $2.11$]{jacod2018}, while Theorem~\ref{maintheo}.\ref{main4} follows by properties of convergence in probability, and the continuous mapping theorem.
\end{proof}
\begin{proof}[\textbf{Proof of Theorem~\ref{maintheo_final}}]
Let $B(\theta, \thetan)$ and $D(\thetan, \thetan)$ be as given in Lemmas \ref{uconvP} and \ref{lemma:uconvP_final}, respectively, and define the block diagonal matrix
  \begin{align*}
    B_0(\theta, \thetan) &= \text{diag}((B_{jk}(\theta, \thetan))_{j,k \in \{1,2\}}, B_{33}(\theta, \thetan))\,.
  \end{align*}
  Note that $B_0(\thetan, \thetan)$ equals $\text{diag}(B_1(\thetan, \thetan), B_2(\thetan, \thetan))$, and that $D_1(\thetan, \thetan)$ is equal to $(C_{jk}(\thetan, \thetan))_{j,k \in \{1,2\}}$ of Lemma~\ref{uconvP}.
  \medskip

  As noted in the proof of Theorem~\ref{maintheo}, $G_n(\theta)$ satisfies \cite[Condition $2.2$]{jacod2018} with $G(\theta) = A(\theta,\thetan)$ for all compact, convex sets $K\subseteq \Theta$ with $\theta_0 \in \text{int}\, K$. Additionally, by Lemmas~\ref{uconvP} and \ref{lemma:uconvP_final}, \cite[Condition $2.10$]{jacod2018} is satisfied with $A_n = \textrm{diag}(\sqrt{n\Delta_n}, \sqrt{n\Delta_n}, \sqrt{n})$, $Z = \nn_3(0,D(\theta_0, \thetan))$, and $W(\theta) = B_0(\theta,\thetan)$, so \eqref{eqn:convD_final} follows from \cite[Theorem $2.11$]{jacod2018}. The proof is completed by application of properties of convergence in probability and the continuous mapping theorem. 
\end{proof}

\subsection{Proofs of Main Lemmas}\label{sec:main_lemma_proofs}
In the following, we use the notation $\EE^n_{i-1}\left( \,\cdot\, \right) = \EE( \,\cdot\,\mid \xtl)$. A martingale difference central limit theorem \cite[Corollary $3.1$]{hall1980} is used several times without reference. 
\begin{proof}[\textbf{Proof of Lemma~\ref{uconvP}}]
The smoothness and polynomial growth assumptions on the integrands of $A(\theta,\thetan)$, $B(\theta,\thetan)$, and $C(\theta,\thetan)$, as well as Assumption \ref{assumptions_on_X_jump}.\ref{x6}, ensure that the results in Lemma  \ref{uconvP}.\ref{ucon1} hold. \medskip

To prove Lemma~\ref{uconvP}.\ref{ucon3}, use Lemmas \ref{lemma:condexpec} and \ref{ucp} to see that for  $j,k =1,\ldots,d$,
\begin{align}
\begin{split} 
&\hspace{-5mm} \frac{1}{n\Delta_n} \sum_{i=1}^n \EE^n_{i-1}\left( g_j(\Delta_n,\xtr,\xtl, \theta)\right)  \\
&= \frac{1}{n} \sum_{i=1}^n \left( \ll( g_j(0, \theta) )(\xtl,\xtl) - \ll_\theta( g_j(0, \theta) )(\xtl, \xtl)\right) \\
&\hspace{5mm} + \Delta_n \frac{1}{n} \sum_{i=1}^n R(\Delta_n, \xtl, \theta)\nonumber\\
&\overset{\pp}{\longrightarrow} A_j(\theta, \thetan)\,,\\
\end{split}\\[0.5em]
\begin{split}
&\hspace{-5mm} \frac{1}{n\Delta_n} \sum_{i=1}^n \EE^n_{i-1}\left( \partial_{\thetak} g_j(\Delta_n,\xtr, \xtl, \theta)\right) \\
&= \frac{1}{n} \sum_{i=1}^n\left( \ll(\partial_{\thetak} g_j(0,\theta))(\xtl,\xtl) - \partial_{\thetak} \ll_\theta(g_j(0, \theta))(\xtl,\xtl)\right) \\
&\hspace{5mm} + \Delta_n\frac{1}{n} \sum_{i=1}^n R(\Delta_n,\xtl, \theta) \\
&\overset{\pp}{\longrightarrow} B_{jk}(\theta, \thetan)\,,\nonumber\\
\end{split}\\[0.5em]
\begin{split}
&\hspace{-5mm} \frac{1}{n\Delta_n} \sum_{i=1}^n \EE^n_{i-1}\left( g_jg_k(\Delta_n,\xtr,\xtl, \theta) \right) \\
&= \frac{1}{n} \sum_{i=1}^n \ll\left(g_jg_k(0,\theta)\right)(\xtl,\xtl) + \Delta_n\frac{1}{n}  \sum_{i=1}^n R(\Delta_n,\xtl, \theta) \\
&\overset{\pp}{\longrightarrow} C_{jk}(\theta, \thetan)\,,\label{EconC}\\
\end{split}
\end{align}
and
\begin{align*}
\begin{split}
\frac{1}{(n\Delta_n)^2}\sum_{i=1}^n \EE^n_{i-1}\left( g_j^2(\Delta_n, \xtr, \xtl, \theta)\right)  = \frac{1}{n\Delta_n} \frac{1}{n}\sum_{i=1}^nR(\Delta_n, \xtl, \theta) 
&\overset{\pp}{\longrightarrow} 0\,,\nonumber
\end{split}\\[0.5em]
\begin{split}
\frac{1}{(n\Delta_n)^2} \sum_{i=1}^n \EE^n_{i-1}\left( (\partial_{\thetak} g_j)^2 (\Delta_n,\xtr, \xtl, \theta)\right) =\frac{1}{n\Delta} \frac{1}{n} \sum_{i=1}^nR(\Delta_n, \xtl, \theta) &\overset{\pp}{\longrightarrow} 0\,,\nonumber
\end{split}\\[0.5em]
\begin{split}
\frac{1}{(n\Delta_n)^2} \sum_{i=1}^n \EE^n_{i-1}\left( g_j^2g_k^2(\Delta_n,\xtr,\xtl,     \theta) \right)= \frac{1}{n\Delta_n} \frac{1}{n}\sum_{i=1}^n R(\Delta_n, \xtl, \theta) &\overset{\pp}{\longrightarrow} 0\,.
\end{split}
\end{align*}
Consequently, by \cite[Lemma 9]{genoncatalot1993}, it holds that pointwise for $\theta \in \Theta$,
\begin{align*}
A^n_j(\theta) \coloneqq \frac{1}{n\Delta_n}\sum_{i=1}^n g_j(\Delta_n, \xtr, \xtl, \theta)  &\overset{\pp}{\longrightarrow} A_j(\theta, \thetan)\,,\nonumber\\
B^n_{jk}(\theta) \coloneqq \frac{1}{n\Delta_n} \sum_{i=1}^n \partial_{\thetak} g_j(\Delta_n,\xtr,
\xtl, \theta) &\overset{\pp}{\longrightarrow} B_{jk}(\theta,\thetan)\,,\nonumber\\
C^n_{jk}(\theta) \coloneqq \frac{1}{n\Delta_n}\sum_{i=1}^n g_jg_k(\Delta_n, \xtr, \xtl, \theta) &\overset{\pp}{\longrightarrow} C_{jk}(\theta,\thetan)\,.
\end{align*}
Let a compact, convex set $K\subseteq \Theta$ be given. The functions $g_j$, $\partial_\thetak g_j$, and $g_jg_k$ satisfy the conditions of Lemma~\ref{for_uni_P}, which may be used, together with Jensen's inequality, the mean value theorem, and the Cauchy-Schwarz inequality, to conclude the existence of constants $p>d$ and $C_{K,p}>0$ such that for all $\theta, \theta' \in K$,
\begin{align*}
  \EE\left( \left|H^n(\theta) - H(\theta, \thetan) - H^n(\theta') +  H(\theta', \thetan)\right|^p\right)
  & \leq C_{K,p} \Vert\theta-\theta'\Vert^p
\end{align*}
for $(H^n, H)$ equal to $(A_j^n, A_j)$, $(B_{jk}^n, B_{jk})$, and $(C_{jk}^n, C_{jk})$, $j,k \in \{1, 2,\ldots, d\}$. Now, using results on convergence in probability from \cite[Chapter 14]{kallenberg1997}, it holds that 
\[
\sup_{\theta \in K}\left\vert H^n(\theta) - H(\theta, \thetan)\right\vert \overset{\pp}{\longrightarrow} 0
\]
as $n\to \infty$, and the results in Lemma \ref{uconvP}.\ref{ucon3} follow. \medskip

In order to prove Lemma~\ref{uconvP}.\ref{ucon2}, suppose first that the estimating function is a martingale estimating function, i.e.\ that $R_\theta(t,x) \equiv 0$ in (\ref{jump_AMG}). Let $v\in \RR^d$ be a fixed vector, and consider
\begin{align}
M_{n,i} &= \frac{1}{\sqrt{n\Delta_n}} \sum_{j=1}^i v^\star g(\Delta_n,\xtrj,\xtlj, \thetan)\,,
\label{M1}
\end{align}
which constitutes a real-valued, zero\--mean, square\--integrable martingale array with martingale differences $D_{n,i} =(n\Delta_n)^{-1/2} v^\star g(\Delta_n,\xtr,\xtl, \thetan)$. By (\ref{EconC}), it holds that
\begin{align*}
\frac{1}{n\Delta_n} \sum_{i=1}^n \EE^n_{i-1}\left( (v^\star g)^2(\Delta_n,\xtr,\xtl, \thetan) \right) &\overset{\pp}{\longrightarrow} v^\star  C(\thetan,\thetan) v\,.
\end{align*}
Furthermore, the conditional Lyapunov condition
\begin{align}
\frac{1}{(n\Delta_n)^2} \sum_{i=1}^n \EE^n_{i-1}\left( (v^\star g)^4(\Delta_n,\xtr,\xtl, \thetan) \right) &\overset{\pp}{\longrightarrow} 0
\label{condlya1_again}
\end{align}
holds, implying the Lindeberg condition. The convergence in (\ref{condlya1_again}) may be verified by multiplying out the parentheses in the left-hand side of the expression and using Lemmas~\ref{lemma:condexpec} and \ref{ucp}. It follows that  
\begin{align*}
\frac{1}{\sqrt{n\Delta_n}} \sum_{i=1}^n v^\star  g(\Delta_n,\xtr,\xtl, \thetan)&\overset{\dd}{\longrightarrow} \nn\left(0,v^\star  C(\thetan,\thetan) v\right)\,,
\end{align*}
and, by the Cram\'{e}r-Wold device, that Lemma~\ref{uconvP}.\ref{ucon2} holds for martingale estimating functions. \medskip 

If the estimating function is not a martingale estimating function, i.e.\ if (\ref{jump_AMG}) holds with $R_\theta(\Delta_n,\xtl)$ not identically  $0$, it is assumed that $n\Delta_n^{2{\kappa_0}-1} \to 0$ as $n\to \infty$ for some ${\kappa_0} \geq 2$. Let
\begin{align}
\tilde{g}(\Delta_n,\xtr,\xtl, \thetan) &= g(\Delta_n,\xtr,\xtl, \thetan) - \EE^n_{i-1}\left(g(\Delta_n,\xtr,\xtl,\thetan)\right)\,.
\label{tildeg}
\end{align}
Since
\begin{align*}
  \frac{1}{\sqrt{n\Delta_n}} \sum_{i=1}^n \EE^n_{i-1}\left(g(\Delta_n,\xtr,\xtl, \thetan)\right) = \sqrt{n} \Delta_n^{{\kappa_0}-1/2} \frac{1}{n}\sum_{i=1}^n R(\Delta_n,\xtl)
  &\overset{\pp}{\longrightarrow} 0
\end{align*}
by (\ref{jump_AMG}) and Lemma \ref{ucp}, it remains to show that
\begin{align}
\frac{1}{\sqrt{n\Delta_n}} \sum_{i=1}^n \tilde{g}(\Delta_n,\xtr,\xtl, \thetan)&\overset{\dd}{\longrightarrow} \nn_d(0,C(\thetan, \thetan))\,.
\label{nonMGconvD}
\end{align}
Consider (\ref{M1}) with $g$ replaced by $\tilde{g}$.
By (\ref{jump_AMG}) and (\ref{EconC}), 
\begin{align*}
&\hspace{-5mm} \frac{1}{n\Delta_n}\sum_{i=1}^n \EE^n_{i-1}\left((v^\star \tilde{g})^2(\Delta_n, \xtr, \xtl, \thetan)\right)\\
&= v^\star \left( \frac{1}{n\Delta_n}\sum_{i=1}^n \EE^n_{i-1}\left(gg^\star (\Delta_n, \xtr, \xtl, \thetan)\right) \right) v \\
&\hspace{5mm} - v^\star  \left( \frac{1}{n\Delta_n}\sum_{i=1}^n  \EE^n_{i-1}\left(g(\Delta_n, \xtr, \xtl, \thetan) \right) \EE^n_{i-1}\left(g(\Delta_n, \xtr, \xtl, \thetan)\right)^\star \right)v\\
&\overset{\pp}{\longrightarrow} v^\star  C(\thetan,\thetan)v\,.
\end{align*}
Furthermore, the conditional Lyapunov condition (\ref{condlya1_again}) with $g$ replaced by $\tilde g$ holds as well, and may be verified in the same manner as (\ref{condlya1_again}). Thus, (\ref{nonMGconvD}) follows, proving Lemma~\ref{uconvP}.\ref{ucon2} also for approximate martingale estimating functions.
\end{proof}
\begin{proof}[\textbf{Proof of Lemma~\ref{lemma:uconvP_final}}]
Let 
\begin{align*}
D_2(\theta, \thetan) &= \int_\xx \tfrac{1}{2}\left( b^4(x, \betan) + \tfrac{1}{2}\left( b^2(x, \betan) - b^2(x,\beta)\right)^2\right)\left(\partial^2_y g_3(0,x,x,\theta)\right)^2\, \pi(dx)\,.
\end{align*}
First, use Lemmas~\ref{lemma:expectations_rate}, \ref{lemma:expectations2_rate}, and \ref{ucp} to see that 
\begin{align}
\begin{split}\label{Eg32}
&\hspace{-5mm} \frac{1}{n\Delta_n^2} \sum_{i=1}^n \EE^n_{i-1}\left( g_3^2(\Delta_n,\xtr,\xtl, \theta)\right)\\
&= \frac{1}{n} \sum_{i=1}^n \tfrac{1}{2} b^4(\xtl, \betan)  \left( \partial^2_y g_3(0,\xtl,\xtl, \theta)\right)^2 \\
&\hspace{5mm} +\frac{1}{n} \sum_{i=1}^n \tfrac{1}{4} \left(b^2(\xtl, \betan)-b^2 (\xtl, \beta)\right)^2   \left(\partial^2_y g_3(0,\xtl,\xtl, \theta)\right)^2 \\
&\hspace{5mm} + \Delta_n\frac{1}{n}\sum_{i=1}^n R(\Delta_n,\xtl, \theta)\\
&\overset{\pp}{\longrightarrow} D_2(\theta, \thetan)\,,
\end{split}
\end{align}
\begin{align}
\frac{1}{n^2\Delta_n^4} \sum_{i=1}^n \EE^n_{i-1}\left( g_3^4(\Delta_n,\xtr,\xtl, \theta)\right) = \frac{1}{n\Delta_n}\frac{1}{n} \sum_{i=1}^n R(\Delta_n,\xtl, \theta)
&\overset{\pp}{\longrightarrow} 0\,,\nonumber
\end{align}
and that for $j=1,2$
\begin{align}
\label{Egjg3}
\frac{1}{n\Delta_n^{3/2}} \sum_{i=1}^n
\EE^n_{i-1}\left( g_jg_3(\Delta_n,\xtr,\xtl, \theta)\right)
= \Delta_n^{1/2} \frac{1}{n} \sum_{i=1}^n R(\Delta_n, \xtr, \theta)
  &\overset{\pp}{\longrightarrow} 0\,, \\
\frac{1}{n\Delta_n^{3/2}} \sum_{i=1}^n
  \EE^n_{i-1}\left( \partial_{\alpha_j} g_3(\Delta_n,\xtr,\xtl, 
\theta)\right) = \Delta_n^{1/2}\frac{1}{n} \sum_{i=1}^n
R(\Delta_n,\xtl, \theta)
&\overset{\pp}{\longrightarrow} 0\,,\nonumber \\
\frac{1}{n^2\Delta_n^{3}} \sum_{i=1}^n \EE^n_{i-1}\left( \left(
\partial_{\alpha_j} g_3(\Delta_n,\xtr,\xtl,
\theta)\right)^2\right) = \frac{1}{n}\frac{1}{n} \sum_{i=1}^n
R(\Delta_n,\xtl, \theta)
&\overset{\pp}{\longrightarrow} 0\,. \nonumber
\end{align}
Thus,
\begin{align}
 \begin{split}
\frac{1}{n\Delta_n^2} \sum_{i=1}^n  g_3^2(\Delta_n,\xtr,\xtl, \theta)
 -D_2(\theta, \thetan) &\overset{\pp}{\longrightarrow} 0\,,
 \\
\frac{1}{n\Delta_n^{3/2}} \sum_{i=1}^n \partial_{\alpha_j} g_3(\Delta_n,\xtr,\xtl,
\theta)&\overset{\pp}{\longrightarrow} 0
\end{split}
         \label{conv_new}
\end{align}
pointwise for $\theta \in \Theta$, by \cite[Lemma
9]{genoncatalot1993}. The functions $\partial_{\alpha_j} g_3$ and $g_3^2$, respectively, satisfy the conditions on $f$ in Lemmas~\ref{for_uni_P_alpha4} and \ref{for_uni_P_alpha4b}. Consequently, Lemma~\ref{lemma:uconvP_final}.\ref{con3_2} follows from (\ref{conv_new}), Jensen's inequality, the mean value theorem, the Cauchy-Schwarz inequality, as well as results on convergence in probability from \cite[Chapter 14]{kallenberg1997}. \medskip

In order to prove Lemma~\ref{lemma:uconvP_final}.\ref{con3_1}, first observe that
\begin{align*}
&\hspace{0mm} \frac{1}{(n\Delta_n)^2}\sum_{i=1}^n \delta_n \EE^n_{i-1}\left( gg^\star (\Delta_n,\xtr, \xtl,
  \thetan)\right) \delta_n\\
&= \EE^n_{i-1}\begin{pmatrix} \displaystyle \frac{1}{n\Delta_n}  \sum_{i=1}^n   g_\alpha   g_\alpha^\star  (\Delta_n,\xtr, \xtl,
  \thetan) & \displaystyle \frac{1}{n\Delta_n^{3/2}} \sum_{i=1}^n  g_\alpha   g_\beta  (\Delta_n,\xtr, \xtl,
  \thetan)\\
\displaystyle \frac{1}{n\Delta_n^{3/2}}   \sum_{i=1}^n   g_\beta g_\alpha^\star  (\Delta_n,\xtr, \xtl,
  \thetan) & \displaystyle \frac{1}{n\Delta_n^2} \sum_{i=1}^n   g_\beta^2  (\Delta_n,\xtr, \xtl,
  \thetan)
\end{pmatrix}   \,,
\end{align*}
so combining (\ref{EconC}) and Remark~\ref{re:jump_gen_presentation} for the submatrix concerning $g_\alpha g_\alpha^\star $, and (\ref{Eg32}) and (\ref{Egjg3}) for the remaining coordinates, it follows that
\begin{align}
\frac{1}{(n\Delta_n)^2}\sum_{i=1}^n \delta_n\EE^n_{i-1}\left( gg^\star (\Delta_n,\xtr, \xtl,
  \thetan)\right) \delta_n
  &\overset{\pp}{\longrightarrow} D(\thetan, \thetan)\,.
\label{ECon}
\end{align}
Suppose now that $R_\theta(t,x) \equiv 0$ in (\ref{jump_AMG}). Let
$v\in \RR^3$ be fixed, and consider
\begin{align}
M_{n,i} &= \frac{1}{n\Delta_n} \sum_{j=1}^i v^\star \delta_n g(\Delta_n,\xtrj,\xtlj,
          \thetan)
          \label{M2}
\end{align}
which defines a real-valued, zero-mean, square-integrable martingale array with differences $D_{n,i}
= (n\Delta_n)^{-1} v^\star \delta_ng(\Delta_n,\xtr,\xtl,
\thetan)$. By (\ref{ECon}), it holds that
\begin{align*}
\sum_{i=1}^n  \EE^n_{i-1}\left(
  \left((n\Delta_n)^{-1}v^\star \delta_ng (\Delta_n,\xtr,\xtl,
\thetan)\right)^2\right)
&\overset{\pp}{\longrightarrow} v^\star  D(\thetan, \thetan) v\,.
\end{align*}
Furthermore, the conditional Lyapunov condition
\begin{align}
\sum_{i=1}^n \EE^n_{i-1}\left( \left((n\Delta_n)^{-1}v^\star \delta_n
    g(\Delta_n,\xtr,\xtl, \thetan)\right)^4 \right)
&\overset{\pp}{\longrightarrow} 0
\label{condlya1_jump}
\end{align}
holds, implying also the Lindeberg condition. The Lyapunov condition may be verified by
  multiplying out the parentheses on the left-hand side of (\ref{condlya1_jump}), and using
  (\ref{Eg4_jump}), and Lemmas~\ref{lemma:expectations2_rate} and \ref{ucp}. It follows then that 
\begin{align*}
\frac{1}{n\Delta_n}\sum_{i=1}^n v^\star \delta_n g(\Delta_n,\xtr,\xtl,
\thetan) &\overset{\dd}{\longrightarrow}\nn\left(0,v^\star  D(\thetan,\thetan) v\right)\,,
\end{align*}
thus proving Lemma~\ref{lemma:uconvP_final}.\ref{con3_1} when $R_\theta(t,x) \equiv 0$ in (\ref{jump_AMG}).
\medskip

When $R_\theta(t,x)$ is not identically equal to $0$ in (\ref{jump_AMG}),
it is assumed that $n\Delta_n^{2({\kappa_0}-1)}
\to 0$ as $n\to \infty$ for some ${\kappa_0} \geq 2$. In this case,
define $\tilde g$ as in (\ref{tildeg}). It holds that
\begin{align*}
\frac{1}{n\Delta_n}\sum_{i=1}^n \delta_n\EE_{i-1}^n \left( g(\Delta_n,\xtr,\xtl,
  \thetan)\right) = \delta_n\Delta_n^{{\kappa_0}-1}\frac{1}{n} \sum_{i=1}^n
                             R(\Delta_n,\xtl)&\overset{\pp}{\longrightarrow} 0\,,
\end{align*}
so it remains to show that
\begin{align}
\frac{1}{n\Delta_n}\sum_{i=1}^n \delta_n\tilde{g}(\Delta_n,\xtr,\xtl,
\thetan)&\overset{\dd}{\longrightarrow} \nn_3(0,D(\thetan, \thetan))\,.
\label{nonMGconvD_final}
\end{align}
Consider
(\ref{M2}) with $\tilde g$ in place of $g$. First, see that
\begin{align*}
&\hspace{-5mm} \sum_{i=1}^n
\EE^n_{i-1}\left(\left((n\Delta_n)^{-1}  v^\star \delta_n \tilde{g}(\Delta_n, \xtr, \xtl, \thetan)\right)^2\right)\\
&= v^\star \left(\frac{1}{(n\Delta_n)^2}\sum_{i=1}^n \delta_n 
\EE^n_{i-1}\left(gg^\star (\Delta_n, \xtr, \xtl, \thetan)\right)
  \delta_n\right)v\\
&\hspace{5mm} - v^\star \delta_n\Delta_n^{{\kappa_0}-1} \left(\frac{1}{n^2}\sum_{i=1}^n 
R(\Delta_n,\xtl)\right)\delta_n\Delta_n^{{\kappa_0}-1} v\\
&\overset{\pp}{\longrightarrow} v^\star  D(\thetan, \thetan)v\,.
\end{align*}
Also, the conditional Lyapunov condition (\ref{condlya1_jump}) holds
with $\tilde g$ in place of $g$. This is  seen by multiplying out the
  parentheses in the conditional expectation, and using Lemmas~\ref{lemma:condexpec},
\ref{lemma:expectations_rate}, \ref{lemma:expectations2_rate}, and Lemma~\ref{ucp}. Now (\ref{nonMGconvD_final}) follows, completing the proof of Lemma~\ref{lemma:uconvP_final}.\ref{con3_1}.
\end{proof}

\appendix

\section{Auxiliary Results}\label{app_lemmas}
This appendix contains technical results pertaining
to the proofs in Section
\ref{sec:main_lemma_proofs}. When not otherwise
mentioned, the general model given by (\ref{jumpSDE}) is assumed.
Appendix \ref{sec:jump:tech:exp} contains inequalities involving expectations, most of them used to prove uniform convergence
in probability, while Appendix \ref{exconmo} concerns expansions
of conditional moments in terms of the infinitesimal generator (\ref{infinitesimalgenerator2}).

\subsection{Inequalities}\label{sec:jump:tech:exp}
\begin{assumption}
The
function $f(t,y,x,\theta)$, $f: (0,\Delta_0)_{\varepsilon_0} \times \xx^2 \times
    \Theta \to \RR$, and its mixed partial derivatives $\partial_t^i
    \partial_y^j f$, $i=0,1$, $j=0,1,2$, exist, are
    continuous, and
  are of polynomial growth in $x$ and $y$, uniformly for $t \in
  (0,\Delta_0)_{\varepsilon_0}$ and $\theta$ in compact, convex sets. Furthermore, $f(0,x,x,\theta)=0$ for all $x\in\xx$ and $\theta\in\Theta$.
\dqed
\label{assumptions_on_f_new}
\end{assumption}
\begin{definition}
 For $q\in \NN$, let $\mathbf{z}_q=(z_1,\ldots,z_q)^\star  \in
\RR^q$ with the convention $\mathbf{z}_0=()$. Let
$m\in\NN_0$ be given. Suppose that Assumption~\ref{assumptions_on_X_jump} holds, and that the function $(t,y,x,\theta) \mapsto f(t,y,x, \mathbf{z}_m,\theta)$
satisfies Assumption
\ref{assumptions_on_f_new}. Define $\aaa_1$, $\aaa_2$, and $\aaa_3$ by their actions on
$f$, which
result in the functions
\begin{align*}
\begin{split}
\aaa_1  f&:\quad (t,y,x, \mathbf{z}_m,\theta)  \mapsto \partial_t
f(t,y,x, \mathbf{z}_m, \theta) + \ll f(t,y,x, \mathbf{z}_m, \theta)\\
\aaa_2 f&: \quad(t,y,x, \mathbf{z}_m,\theta) \mapsto b(y, \thetan) \partial_y f(t,y,x, \mathbf{z}_m, \theta)\\
\aaa_3 f&:\quad (t,y,x,
\mathbf{z}_{m+1},\theta) \mapsto f(t,y+c(y,z_{m+1}, \thetan), x, \mathbf{z}_m, \theta) - f(t,y,x, \mathbf{z}_m, \theta)\,.
\end{split}
\end{align*}
\dqed
\end{definition}
\begin{remark}
When well-defined for $f(t,y,x, \mathbf{z}_m,
\theta)$, the notation $f_j \coloneqq\aaa_j f$, $f_{jk}\coloneqq\aaa_k\aaa_j
f$, and $f_{j3k} \coloneqq \aaa_k\aaa_3\aaa_j f$ is used for
$j,k=1,2,3$. We put $\bm{h}(u,s,  \mathbf{z}_m,
  \theta) \coloneqq h(u-s,X_u,X_s, \mathbf{z}_m, \theta)$, as well as $\bm{h}(u-,s,  \mathbf{z}_m,
  \theta) \coloneqq h(u-s,X_{u-},X_s, \mathbf{z}_m, \theta)$, and $D\bm{h}(\,\cdot\,, \theta,\theta') \coloneqq \bm{h}(\,\cdot\,,
\theta)-\bm{h}(\,\cdot\,, \theta')$ for functions $h (t,y,x,\mathbf{z}_m,
  \theta)$. \cqed
\label{note:f_for_Ito}
\end{remark}
\begin{lemma}
Suppose that Assumption~\ref{assumptions_on_X_jump} holds, and that $f(t,y,x, \theta)$ satisfies
Assumption~\ref{assumptions_on_f_new}. Let $p = 2^q$ for
some $q \in \NN$, and recall the notation from Remark
\ref{note:f_for_Ito}. 
Then, there exist constants
$C_p>0$  such that for all $\theta,\theta' \in \Theta$ and $n\in\NN$,
\begin{align*}
\begin{split}
&\hspace{-5mm} \EE\left( \left| \sum_{i=1}^n D\bm{f}(\tplus, \tminus,\theta,\theta') \right|^p\right)\\
&\leq  (n\Delta_n)^{p-1} C_p\sum_{i=1}^n\int_\tminus^\tplus
  \EE\left( D\bm{f_1}(u,\tminus,\theta,\theta')^p\right) \, du\\
&\hspace{5mm} + (n\Delta_n)^{p/2-1} C_p \sum_{i=1}^n
 \int_\tminus^\tplus  \EE\left(
D\bm{f_2}(u,\tminus,\theta,\theta')^{p}\right)\, du\\
&\hspace{5mm} + \sum_{l=1}^q (n\Delta_n)^{2^{q-l}-1}
 C_p \sum_{i=1}^n \int_\tminus^\tplus \int_\RR  \EE\left(
    D\bm{f_3}(u,\tminus,z, \theta,\theta')^{p}\right)
  \nu_\thetan(dz)\, du\,.
\end{split}
\end{align*}
\dqed
\label{ito_ineq_jumps_sum}
\end{lemma}
\begin{proof}[\textbf{Proof of Lemma~\ref{ito_ineq_jumps_sum}}]
  By It\^{o}'s formula for SDEs with jumps,
  \begin{align}
&\hspace{-5mm} \EE\left( \left| \sum_{i=1}^n D\bm{f}(\tplus, \tminus,\theta,\theta') \right|^p\right)\nonumber\\
&\leq C_p \,\EE\left( \left|\sum_{i=1}^n \int_\tminus^\tplus D\bm{f_1}(u-,\tminus,
  \theta,\theta')\, du\right|^p\right) \label{ito_sum_1}\\
&\hspace{5mm} + C_p \,\EE\left( \left|\sum_{i=1}^n \int_\tminus^\tplus
D\bm{f_2}(u-,\tminus, \theta, \theta')\, dW_u\right|^p\right)\label{ito_sum_2}\\
&\hspace{5mm} + C_p \,\EE\left( \left|\sum_{i=1}^n \int_\tminus^\tplus \int_\RR
D\bm{f_3}(u-,\tminus,z, \theta, \theta')\, (N^\thetan-\mu_\thetan)(du,dz)\right|^p\right)\label{ito_sum_3}
  \end{align}
for suitable constants $C_p>0$. Starting with (\ref{ito_sum_1}), using Jensen's inequality twice,
\begin{align}
&\hspace{-5mm} \EE\left( \left|\sum_{i=1}^n \int_\tminus^\tplus D\bm{f_1}(u-,\tminus,\theta,\theta')\, du\right|^p\right)\nonumber\\
&= (n\Delta_n)^p\, \EE\left( \left|\frac{1}{n}\sum_{i=1}^n \frac{1}{\Delta_n}\int_\tminus^\tplus D\bm{f_1}(u-,\tminus, \theta,\theta')\, du\right|^p\right)\nonumber\\
&\leq (n\Delta_n)^{p-1} \sum_{i=1}^n\int_\tminus^\tplus
  \EE\left(  D\bm{f_1}(u-,\tminus,  \theta,\theta')^p\right) \, du\,.\label{ito_sum_1_new}
\end{align}
Now, consider (\ref{ito_sum_2}). In the same manner as before, using also the Burk\-hol\-der-Da\-vis-Gun\-dy inequality,
\begin{align}
&\hspace{-5mm} \EE\left( \left|\sum_{i=1}^n \int_\tminus^\tplus
D\bm{f_2}(u-,\tminus,\theta, \theta')\,
  dW_u\right|^p\right)\nonumber\\
&=(n\Delta_n)^p \,\EE\left( \left|\frac{1}{n} \sum_{i=1}^n \frac{1}{\Delta_n}\int_\tminus^\tplus
  D\bm{f_2}(u-,\tminus,\theta, \theta')\,
  dW_u\right|^p\right)\nonumber\\
&\leq (n\Delta_n)^p C_p \,\EE\left( \left|\frac{1}{n^2}\sum_{i=1}^n \frac{1}{\Delta_n^2} \int_\tminus^\tplus 
D\bm{f_2}(u-,\tminus, \theta,\theta')^2\,
  du\right|^{p/2}\right)\nonumber\\
&\leq (n\Delta_n)^{p/2-1} C_p \sum_{i=1}^n
 \int_\tminus^\tplus  \EE\left(
D\bm{f_2}(u-,\tminus, \theta,\theta')^{p}\right)\, du\,,
\label{ito_sum_2_new}
\end{align}
for some constant $C_p > 0$. Finally, in the case of (\ref{ito_sum_3}), let
$\M^{(k)} = (M_v^{(k)})_{v\geq 0}$ and $\S^{(k)} = (S_v^{(k)})_{v\geq 0}$ be given by
\begin{align}
  \begin{split}
M_v^{(k)} &= \int_0^v \int_\RR \sum_{i=1}^n\mathbf{1}_{(\tminus,\tplus]}(u)
            D\bm{f_3}(u-,\tminus,z, \theta,\theta')^k (N^\thetan-\mu_\thetan)(du,dz) \\
S_v^{(k)} &= \int_0^v \int_\RR \sum_{i=1}^n
\mathbf{1}_{(\tminus,\tplus]}(u) D\bm{f_3}(u-,\tminus,z,
  \theta,\theta')^k \, \nu_\thetan(dz)\, du
\end{split}
\label{MSdef}
\end{align}
for $k\in \NN$, and note that the quadratic variation of $\M^{(k)}$ may be written as
\begin{align*}
\int_0^v \int_\RR \sum_{i=1}^n\mathbf{1}_{(\tminus,\tplus]}(u) D\bm{f_3}(u-,\tminus,z, \theta,\theta')^{2k} \, N^\thetan(du,dz) 
&= M_v^{(2k)} + S_v^{(2k)}\,.
\end{align*}
By the Burkholder-Davis-Gundy inequality, for any $m\geq
1$, there exist constants $C_m>0$ such that
\begin{align*}
\EE\left( |M_v^{(k)}|^m\right) &\leq C_m  \,\EE\left( \left(M_v^{(2k)}\right)^{m/2}\right) + C_m \,\EE\left( \left(S_v^{(2k)}\right)^{m/2}\right)\,.
\end{align*}
In particular, inserting $2^j$ in place of $k$ and $2^{q-j}$ in
place of $m$ for $j \in \{0,1,\ldots,q-1\}$,
\begin{align*}
\EE\left( \left(M_v^{(2^j)}\right)^{2^{q-j}}\right) 
&\leq C_p \, \EE\left(
  \left(M_v^{(2^{j+1})}\right)^{2^{q-(j+1)}}\right) + C_p\,
\EE\left(
  \left(S_v^{(2^{j+1})}\right)^{2^{q-(j+1)}}\right)\,.
\end{align*}
This inequality may be used iteratively to obtain
\begin{align*}
\EE\left( \left(M_v^{(1)}\right)^p\right)
&\leq 
C_p \sum_{l=1}^q \EE\left(
  \left(S_v^{(2^l)}\right)^{2^{q-l}}\right)\,,
\end{align*}
where we used that $\EE( M_v^{(p)} ) =0$ by properties
of the Poisson integral.
Inserting from (\ref{MSdef}), this may also be written as
\begin{align}
  \begin{split}
&\hspace{-5mm} \EE\left( \left|\sum_{i=1}^n \int_\tminus^\tplus \int_\RR
  D\bm{f_3}(u-,\tminus,z, \theta, \theta')\,
  (N^\thetan-\mu_\thetan)(du,dz)\right|^p\right)\\
&\leq C_p \sum_{l=1}^q \EE\left(\left( \sum_{i=1}^n \int_\tminus^\tplus \int_\RR
D\bm{f_3}(u-,\tminus,z, \theta, \theta')^{2^l}\,
 \nu_\thetan(dz)\,du\right)^{2^{q-l}}\right)\,.
\end{split}
\label{new_ineq}
\end{align}
Recalling that $\nu_\theta$ has density $\xi(\theta)p(\,\cdot\,,\theta)$ with
respect to $\tilde{\nu}$, where $p(\,\cdot\,,\theta)$ is a probability
density, Jensen's inequality is used twice to write
\begin{align*}
&\hspace{-5mm} \EE\left(
  \left(\sum_{i=1}^n \int_\tminus^\tplus \int_\RR  D\bm{f_3}(u-,\tminus,z, \theta,\theta')^{2^l} \, \nu_\thetan(dz)\,du\right)^{2^{q-l}}\right)\nonumber\\
&= (\xi(\thetan)n\Delta_n)^{2^{q-l}} \\
&\hspace{5mm} \times \EE\left(
  \left(\frac{1}{n} \sum_{i=1}^n \frac{1}{\Delta_n}\int_\tminus^\tplus \int_\RR 
    D\bm{f_3}(u-,\tminus,z, \theta,\theta')^{2^l} p(z,\thetan)\,
    \tilde{\nu}(dz)\, du\right)^{2^{q-l}}\right) \nonumber\\
&\leq  (\xi(\thetan)n\Delta_n)^{2^{q-l}-1}
  \sum_{i=1}^n \int_\tminus^\tplus \int_\RR  \EE\left(
    D\bm{f_3}(u-,\tminus,z, \theta,\theta')^{p}\right)
  \nu_\thetan(dz)\, du\nonumber\\
&= (n\Delta_n)^{2^{q-l}-1}C_p
  \sum_{i=1}^n \int_\tminus^\tplus \int_\RR  \EE\left(
    D\bm{f_3}(u-,\tminus,z, \theta,\theta')^{p}\right)
  \,\nu_\thetan(dz)\, du\,.
\end{align*}
Inserting this into (\ref{new_ineq}), we obtain
\begin{align}
&\hspace{5mm} \EE\left( \left|\sum_{i=1}^n \int_\tminus^\tplus \int_\RR
 D\bm{f_3}(u-,\tminus,z, \theta, \theta')\,
  (N^\thetan-\mu_\thetan)(du,dz)\right|^p\right)\nonumber\\
&\leq \sum_{l=1}^q (n\Delta_n)^{2^{q-l}-1}C_p 
  \sum_{i=1}^n \int_\tminus^\tplus \int_\RR  \EE\left(
D\bm{f_3}(u-,\tminus,z, \theta,\theta')^{p}\right)
  \nu_\thetan(dz)\, du\,.
\label{ito_sum_3_new}
\end{align}
The proof is completed by replacing (\ref{ito_sum_1}), (\ref{ito_sum_2}), (\ref{ito_sum_3}) with (\ref{ito_sum_1_new}), (\ref{ito_sum_2_new}), 
(\ref{ito_sum_3_new}), and using that $\X$ has finite
activity jumps to replace $X_{u-}$ with $X_u$ in the integrals.
\end{proof}
\begin{lemma}
Let $m\in \NN_0$, $p = 2^q$ for
some $q \in \NN$, and recall the notation of Remark
\ref{note:f_for_Ito}. Suppose that
Assumption~\ref{assumptions_on_X_jump} holds. Assume that $(t,y,x, \theta) 
\mapsto f(t,y,x,\mathbf{z}_m, \theta)$ satisfies
Assumption~\ref{assumptions_on_f_new}.
Then, there exist constants
$C_p>0$  such that
\begin{align*} 
\begin{split}
&\hspace{-5mm} \EE\left( D\bm{f}(t,s,\mathbf{z}_m,
  \theta,\theta')^p\right)\\
&\leq(t-s)^{p-1}  C_p\int_s^t \EE\left( 
  D\bm{f_1}(u,s,\mathbf{z}_m, \theta,\theta')^p\right)\, du\\
&\hspace{5mm} + (t-s)^{p/2-1}  C_p\int_s^t \EE\left( 
  D\bm{f_2}(u,s,\mathbf{z}_m, \theta,\theta')^p\right)\, du \\
&\hspace{5mm} +\left( \sum_{l=1}^{q}
  (t-s)^{2^{q-l}-1}\right)  C_p\int_s^t \int_\RR
  \EE\left(D\bm{f_3}(u,s,\mathbf{z}_m,z,
  \theta,\theta')^p\right)\, \nu_\thetan(dz)\, du
\end{split}
\end{align*}
for all $\theta,\theta' \in \Theta$, $0\leq s < t \leq
s+\Delta_0$.  \dqed
\label{ito_ineq_jumps}
\end{lemma}
Letting $f$ depend on an extra
variable $\mathbf{z}_m$ in the proof of Lemma
\ref{ito_ineq_jumps_sum}, and putting $n=1$, $t_i^n = t$ and
$t_{i-1}^n = s$ (so that $\Delta_n = t-s$) proves Lemma~\ref{ito_ineq_jumps}.
\begin{lemma}
Let $p > d$ of the form $p=2^q$ for some $q \in \NN$ be given. Suppose that Assumption
\ref{assumptions_on_X_jump} holds, and that $f(t,y,x, \theta) \in
\cc^\infty_\text{pol}((0,\Delta_0)_{\varepsilon_0} \times \xx^2 \times
\Theta)$
with $f(0,x,x, \theta) =0$ for all $x\in \xx$ and
$\theta \in \Theta$. Let
\begin{align*}
\zeta_n(\theta) &= \frac{1}{n\Delta_n} \sum_{i=1}^n f(\Delta_n, \xtr, \xtl, \theta)\,.
\end{align*}
Then, for each compact, convex set $K\subseteq \Theta$, there exists $C_{K,p}>0$ such that 
\begin{align*}
\EE\left( |\zeta_n(\theta)-\zeta_n(\theta')|^p\right)
&\leq C_{K,p} \Vert \theta-\theta' \Vert^{p}
\end{align*}
for all $\theta, \theta' \in K$ and $n \in \NN$.\dqed
\label{for_uni_P}
\end{lemma}
\begin{proof}[\textbf{Proof of Lemma~\ref{for_uni_P}}]
Recall the notation from Remark~\ref{note:f_for_Ito}. Let $K\subseteq \Theta$ compact and convex be given. Write
\begin{align}
\EE\left( |\zeta_n(\theta)-\zeta_n(\theta')|^p\right) &=
                                                                     (n\Delta_n)^{-p}
                                                                     \,\EE\left(\left|\sum_{i=1}^n
                                                                     D
                                                                     \bm{f}(\tplus,
                                                        \tminus, \theta, \theta')\right|^p\right)\,.
\label{general_desired_result_new}
\end{align}
By Lemma~\ref{ito_ineq_jumps_sum}, there exist constants
$C_p>0$  such that for all $\theta,\theta' \in K$ and $n\in\NN$, 
\begin{align}
\begin{split}
&\hspace{-5mm} \EE\left( \left| \sum_{i=1}^n D\bm{f}(\tplus,\tminus,\theta,\theta') \right|^p\right)\label{general_first_step_new}\\
&\leq  (n\Delta_n)^{p-1} C_p\sum_{i=1}^n\int_\tminus^\tplus
  \EE\left( D\bm{f_1}(u,\tminus,\theta,\theta')^p\right) \, du\\
&\hspace{5mm} + (n\Delta_n)^{p/2-1} C_p \sum_{i=1}^n
 \int_\tminus^\tplus  \EE\left(
D\bm{f_2}(u,\tminus, \theta,\theta')^{p}\right)\, du\\
&\hspace{5mm} + \sum_{l=1}^q (n\Delta_n)^{2^{q-l}-1}
 C_p \sum_{i=1}^n \int_\tminus^\tplus \int_\RR  \EE\left(
    D\bm{f_3}(u,\tminus,z, \theta,\theta')^{p}\right)
  \nu_\thetan(dz)\, du\,.
\end{split}
\end{align}
The mean value
theorem and the Cauchy-Schwarz inequality may be used  to show that there exist constants $C_{K,p}>0$
such that for $j=1,2$,
\begin{align}
\int_\tminus^\tplus \EE\left( D\bm{f_j}(u,\tminus, \theta,
  \theta')^p\right) \, du
&\leq C_{K,p} \,\Delta_n\,\Vert
                      \theta-\theta'\Vert^p 
\label{first_general_bound}\\
  \int_\tminus^\tplus \int_\RR \EE\left( D\bm{f_3}(u,\tminus, z, \theta,
  \theta')^p\right) \,\nu_\thetan(dz) \, du
&\leq C_{K,p} \,\Delta_n\,\Vert
                      \theta-\theta'\Vert^p \,.
\label{second_general_bound}
\end{align}
Inserting (\ref{first_general_bound}) and (\ref{second_general_bound})
into (\ref{general_first_step_new}) yields the existence of
$C_{K,p}>0$ such that
\begin{align}
&\hspace{-5mm} \EE\left( \left| \sum_{i=1}^n D\bm{f}(\tplus,\tminus,\theta,\theta') \right|^p\right)\nonumber\\
&\leq C_{K,p}\left(  (n\Delta_n)^{p} +
  (n\Delta_n)^{p/2} + \sum_{l=1}^q
  (n\Delta_n)^{2^{q-l}} \right) \, \Vert
  \theta-\theta'\Vert^p \nonumber\\
&\leq  C_{K,p}\,(n\Delta_n)^{p}\, \Vert
  \theta-\theta'\Vert^p\,,\label{first_use_of_ii_new}
\end{align}
since $n\Delta_n \to \infty$ as $n\to \infty$. Inserting
(\ref{first_use_of_ii_new}) into
(\ref{general_desired_result_new}) completes the proof.
\end{proof}
\begin{lemma}
Consider the model given by (\ref{eqn:SDEoptimal2}), with $A\subseteq \RR^2$ and $B\subseteq \RR$. Suppose that Assumption
\ref{assumptions_on_X_jump} holds, and that $f(t,y,x, \theta) \in
\cc^\infty_\text{pol}((0,\Delta_0)_{\varepsilon_0} \times \xx^2 \times
\Theta)$. Furthermore, assume that
\begin{align*}
f(0,y,x, \theta) &=0\,, \quad y \in {\cal M}_k(x,\tilde \alpha), \quad k=0,1,2\\
\partial_t f(0,y,x, \theta) &=0\,,\quad y \in {\cal M}_k(x,\tilde \alpha), \quad 
                                                           k=0,1\\
\partial_y f(0,y,x, \theta) &=0\,,\quad y \in {\cal M}_k(x,\tilde \alpha), \quad 
                                                           k=0,1\\
\partial_y^2 f(0,y,x, \theta) &=0\,,\quad y \in {\cal M}_k(x,\tilde
                                \alpha), \quad  k=0,1
\end{align*}
for all $\tilde \alpha \in A$, $\theta \in \Theta$, and $x\in \xx$,
where ${\cal M}_k(y,\alpha)$ is as defined in Section \ref{onedimdiff}.
Let 
\begin{align*}
\zeta_n(\theta) &= \frac{1}{n\Delta_n^{3/2}} \sum_{i=1}^n f(\Delta_n, \xtr, \xtl, \theta)\,.
\end{align*}
Then, for any compact, convex $K\subseteq \Theta$, there exists
$C_{K}>0$, so for all $\theta, \theta' \in K$, $n \in \NN$,
\begin{align*}
\EE\left( |\zeta_n(\theta)-\zeta_n(\theta')|^4\right) &\leq C_{K} \Vert \theta-\theta' \Vert^4\,.
\end{align*}
\dqed
\label{for_uni_P_alpha4}
\end{lemma}
\begin{proof}[\textbf{Proof of Lemma~\ref{for_uni_P_alpha4}}]
Recall the notation of Remark
\ref{note:f_for_Ito}, and note that for $j=1,2,3$, $f$, $f_j$, and $f_{j3}$, as functions of $(t,y,x,\theta)$,
  satisfy Assumption~\ref{assumptions_on_f_new}. Write
\begin{align}
\EE\left( |\zeta_n(\theta)-\zeta_n(\theta')|^4\right) &=
                                                                     (n\Delta_n)^{-4}\Delta_n^{-2}
                                                                     \,\EE\left(\left|\sum_{i=1}^n
                                                                     D
                                                                     \bm{f}(\tplus,\tminus, \theta,
                                                                     \theta')\right|^4\right)\,.
\label{general_desired_result_alpha4}
\end{align}
By Lemma~\ref{ito_ineq_jumps_sum}, there exist constants
$C>0$  such that for all $\theta,\theta' \in \Theta$ and $n\in\NN$,
\begin{align}
\begin{split}
&\hspace{-5mm} \EE\left( \left| \sum_{i=1}^n D\bm{f}(\tplus, \tminus,\theta,\theta') \right|^4\right)\label{general_first_step_alpha4}\\
&\leq  (n\Delta_n)^{3} C\sum_{i=1}^n\int_\tminus^\tplus
  \EE\left( D\bm{f_1}(u, \tminus, \theta,\theta')^4\right) \, du \\
&\hspace{5mm} + n\Delta_n C \sum_{i=1}^n
 \int_\tminus^\tplus  \EE\left(
D\bm{f_2}(u,\tminus, \theta,\theta')^4\right)\, du \\
&\hspace{5mm} + \left( 1 + n\Delta_n\right)
 C \sum_{i=1}^n \int_\tminus^\tplus \int_\RR  \EE\left(
    D\bm{f_3}(u,\tminus,z_1, \theta,\theta')^{4}\right)
  \nu_\alphan(dz_1)\, du \,.
\end{split}
\end{align}
Furthermore,
applying Lemma
\ref{ito_ineq_jumps} twice, there exist constants $C>0$ such that
\begin{align}
\begin{split}
\label{first_part_alpha4}
&\hspace{-4mm} \EE\left( D\bm{f_j}(u ,\tminus,
  \theta,\theta')^4\right)  \\
&\leq C({u}-\tminus)^3\int_\tminus^{u}\EE\left( D\bm{f_{j1}}({v},\tminus,
  \theta,\theta')^4\right)\, d{v} \\
&\hspace{5mm} +  C({u}-\tminus) \int_\tminus^{u} \EE\left( D\bm{f_{j2}}({v},\tminus,
  \theta,\theta')^4\right)\, d{v}\\
&\hspace{5mm} + C\left( 1+ {u}-\tminus\right) \\
&\hspace{10mm} \times \left(\int_\tminus^{u} \int_\RR
  ({v}-\tminus)^3 \int_\tminus^{v} \EE\left(
  D\bm{f_{j31}}({w},\tminus,z_1, \theta,\theta')^4 \right)\,d{w}\,\nu_\alphan(dz_1)\, d{v}\right.\\
&\hspace{15mm} + \int_\tminus^{u} \int_\RR 
({v}-\tminus) \int_\tminus^{v} \EE\left(
  D\bm{f_{j32}}({w},\tminus,z_1, \theta,\theta')^4 \right)\,d{w}
\,\nu_\alphan(dz_1)\, d{v}\\
&\hspace{15mm} + \int_\tminus^{u} \int_\RR
  (1+{v}-\tminus) \\
&\hspace{20mm} \times \left.\int_\tminus^{v} \int_\RR
\EE\left( D\bm{f_{j33}}({w},\tminus,\mathbf{z}_2,
  \theta,\theta')^4\right) \nu_\alphan(dz_2)\, d{w}\,
 \,\nu_\alphan(dz_1)\, d{v}\right)
\end{split}
\end{align}
for $j=1,2$, and
\begin{align}
\begin{split}\label{second_part_alpha4}
&\hspace{-4mm} \EE\left( D\bm{f_3}({u},\tminus,z_1,
  \theta,\theta')^4\right)  \\
&\leq C({u}-\tminus)^3\int_\tminus^{u} \EE\left( D\bm{f_{31}}({v},\tminus,z_1,
  \theta,\theta')^4\right)\, d{v} \\
&\hspace{5mm} +  C({u}-\tminus) \int_\tminus^{u} \EE\left( D\bm{f_{32}}({v},\tminus,z_1,
  \theta,\theta')^4\right)\, d{v}\\
&\hspace{5mm} + C\left( 1+ {u}-\tminus\right) \\
&\hspace{10mm} \times \left( \int_\tminus^{u} \int_\RR
  ({v}-\tminus)^3 \int_\tminus^{v} \EE\left(
  D\bm{f_{331}}({w},\tminus,\mathbf{z}_2, \theta,\theta')^4 \right)\,d{w}\,\nu_\alphan(dz_2)\, d{v}\right.\\
&\hspace{15mm} + \int_\tminus^{u} \int_\RR 
({v}-\tminus) \int_\tminus^{v} \EE\left(
  D\bm{f_{332}}({w},\tminus,\mathbf{z}_2, \theta,\theta')^4 \right)\,d{w}
\,\nu_\alphan(z_2)\, d{v}\\
&\hspace{15mm}+ \int_\tminus^{u} \int_\RR
  (1+{v}-\tminus) \\
&\hspace{20mm}\times\left. \int_\tminus^{v}\int_\RR
\EE\left( D\bm{f_{333}}({w},\tminus,\mathbf{z}_3,
  \theta,\theta')^4\right)\nu_\alphan(dz_3)\, d{w}\,
 \,\nu_\alphan(dz_2)\, d{v}\right)\,.
\end{split}
\end{align}
Let a compact and convex subset $K\subseteq \Theta$ be given. Using
  the mean value theorem and the Cauchy-Schwarz inequality, it may be
  shown that there exist constants $C_K>0$ such that
for $i=1,\ldots,n$, and $j \in \{11, 12, 21, 22 \}$, $ k \in \{31, 32, 131, 132, 231,
  232\}$, and $l \in \{133, 233, 331, 332 \}$,
\begin{align}
\begin{split}\label{final_bound_for_alpha4}
  \EE\left( D\bm{f_j}({w},\tminus,
  \theta,\theta')^4\right)
&\leq C_{K}\,\Vert
  \theta-\theta'\Vert^4\\
\EE\left( D\bm{f_k}({w},\tminus,z_1,
  \theta,\theta')^4\right) 
&\leq C_{K}\, \Vert
  \theta-\theta'\Vert^4 \left( 1+ |z_1|^{C_K}\right)\\
 \EE\left( D\bm{f_l}({w},\tminus,\mathbf{z}_2,
  \theta,\theta')^4\right)
&\leq C_{K}\,\Vert
  \theta-\theta'\Vert^4 \left( 1+ |z_1|^{C_K}\right)\left( 1+ |z_2|^{C_K}\right)\\
  \EE\left( D\bm{f_{333}}({w},\tminus, \mathbf{z}_3,
  \theta,\theta')^4\right) 
&\leq C_{K}\,\Vert
  \theta-\theta'\Vert^4\left( 1+ |z_1|^{C_K}\right)\left( 1+ |z_2|^{C_K}\right)\left( 1+ |z_3|^{C_K}\right).
\end{split}
\end{align}
Inserting (\ref{final_bound_for_alpha4}) into
(\ref{first_part_alpha4}) and (\ref{second_part_alpha4}), it follows that for $j=1,2$,
\begin{align}
\begin{split}
\label{second_part_alpha4_bound}
\EE\left( D\bm{f_j}({u},\tminus,
  \theta,\theta')^4\right)  
&\leq C_K ({u}-\tminus)^2 \,\Vert
\theta-\theta'\Vert^4\\
\EE\left( D\bm{f_3}({u},\tminus, z_1,
  \theta,\theta')^4\right)  
&\leq C_K ({u}-\tminus)^2 \left( 1+ |z_1|^{C_K}\right)\,\Vert
\theta-\theta'\Vert^4\,.
\end{split}
\end{align}
Now, inserting (\ref{second_part_alpha4_bound})
into (\ref{general_first_step_alpha4}) yields the existence of
$C_{K}>0$ such that
\begin{align}
\EE\left( \left| \sum_{i=1}^n D\bm{f}(\tplus,\tminus,\theta,\theta')
  \right|^4\right)\label{general_first_step_alpha4_tool} &\leq C_K (n\Delta_n)^4\Delta_n^2 \,\Vert
  \theta-\theta'\Vert^4\,,
\end{align}
where we used that $n\Delta_n \to \infty$ as $n\to \infty$. Inserting
(\ref{general_first_step_alpha4_tool}) into (\ref{general_desired_result_alpha4}),
the desired result is obtained. 
\end{proof}
The proof of the following Lemma~\ref{for_uni_P_alpha4b} is very similar to the proof of Lemma
\ref{for_uni_P_alpha4}, but requires more applications of Lemma
\ref{ito_ineq_jumps} in order to achieve appropriate orders of
$\Delta_n$. We refer to \cite[Section 3.A.3]{phdthesis} for the
details.
\begin{lemma}
Consider the model given by (\ref{eqn:SDEoptimal2}), with $A\subseteq\RR^2$ and $B\subseteq \RR$. Suppose that Assumption
\ref{assumptions_on_X_jump} holds,
and that $f(t,y,x, \theta) \in
\cc^\infty_\text{pol}((0,\Delta_0)_{\varepsilon_0} \times \xx^2 \times
\Theta)$. Furthermore, assume that
\begin{align*}
f(0,y,x, \theta) &=0\,, \quad y \in {\cal M}_k(x,\tilde \alpha), \quad k=0,1,2,3,4\\
\partial_t f(0,y,x, \theta) &=0\,,\quad y \in {\cal M}_k(x,\tilde \alpha), \quad 
                                                           k=0,1,2,3\\
\partial_y f(0,y,x, \theta) &=0\,,\quad y \in {\cal M}_k(x,\tilde \alpha), \quad 
                                                           k=0,1,2,3\\
\partial_y^2 f(0,y,x, \theta) &=0\,,\quad y \in {\cal M}_k(x,\tilde \alpha), \quad 
                                                             k=0,1,2,3\\
\partial_y^3 f(0,y,x, \theta) &=0\,,\quad y \in {\cal M}_k(x,\tilde \alpha), \quad 
                                                             k=0,1\\
\partial_t\partial_y f(0,y,x, \theta) &=0\,,\quad y \in {\cal M}_k(x,\tilde \alpha), \quad 
                                                           k=0,1\,.
\end{align*}
for all $\tilde \alpha \in A$, $\theta \in \Theta$, and $x\in \xx$, 
where ${\cal M}_k(y,\alpha)$ is as defined in Section
\ref{onedimdiff}. Let
\begin{align*}
\zeta_n(\theta) &= \frac{1}{n\Delta_n^2} \sum_{i=1}^n f(\Delta_n, \xtr, \xtl, \theta)\,.
\end{align*}
Then, for any compact, convex set $K\subseteq \Theta$, there exists a constant $C_{K}>0$ such that 
\begin{align*}
\EE\left( |\zeta_n(\theta)-\zeta_n(\theta')|^4\right)
&\leq C_{K} \Vert \theta-\theta' \Vert^4
\end{align*}
for all $\theta, \theta' \in K$, and $n \in \NN$.\dqed
\label{for_uni_P_alpha4b}
\end{lemma}

\subsection{Expansion of Conditional Moments}\label{exconmo}
\begin{remark}
Note that under Assumptions~\ref{assumptions_on_X_jump} and
\ref{assumptions_on_g_jump},
\begin{align*}
\begin{split}
& \hspace{-5mm} \partial_\theta \ll_\theta(g(0,\theta))(x,x) \\
&= \ll_\theta(\partial_{\theta} g(0,\theta))(x,x) + \partial_y
g(0,x,x, \theta)\partial_\theta  a(x, \theta) + \tfrac{1}{2}\partial^2_y g(0,x,x, \theta)\partial_\theta b^2(x,
\theta)\\
&\hspace{5mm} + \int_\RR \partial_y g(0,x+c(x,z,
\theta), x, \theta) \partial_\theta c(x,z, \theta)\, \nu_\theta(dz) \\
&\hspace{5mm} + \int_\RR g(0,x+c(x,z,
\theta), x, \theta)\partial_\theta q(z,\theta)\, \tilde{\nu}(dz)
\end{split}\\[0.5em]
\begin{split}
&\hspace{-5mm} \ll_\lambda(gg^\star (0,\theta))(x,x) \\
&= b^2(x, \lambda) \partial_y g(\partial_y g)^\star (0,x,x,
  \theta) + \int_\RR gg^\star  (0,x+c(x,z, \lambda),x, \theta) \, \nu_\lambda(dz)
\end{split}
\end{align*}
for all $x\in\xx$ and $\lambda, \theta \in
\Theta$, by (\ref{infinitesimalgenerator2}) and Lemma
\ref{lemma:conseq}. \cqed
\label{re:jump_gen_presentation}
\end{remark}
Furthermore, note that under
Assumption~\ref{assumptions_on_X_jump}, it holds that
\begin{align}
\EE\left( R_{\lambda}(\Delta, X^\lambda_{t+\Delta}, X^\lambda_t,
      \theta) \mid X^\lambda_{t}\right) &=R_{\lambda}(\Delta, X^\lambda_{t}, \theta)
\label{eqn:expectation_remainder_jump}
\end{align}
for $0\leq t
\leq t+\Delta$ with $\Delta \leq \Delta_0$, and $\lambda
\in \Theta$. This follows from a version of \cite[Proposition 3.1]{shimizu2006}, which may also be shown to hold in the current framework.
\begin{lemma}
Suppose that Assumptions~\ref{assumptions_on_X_jump} and
\ref{assumptions_on_g_jump} hold. Then,
\begin{align}
\begin{split}\nonumber
& \hspace{-5mm} \EE^n_{i-1}\left( g(\Delta_n,\xtr, \xtl, \theta) \right)
\\
&= \Delta_n\left( \ll( g(0, \theta) )(\xtl,
  \xtl) - \ll_{\theta}( g(0, \theta) )(\xtl,
  \xtl)\right) \\
&\hspace{5mm}+ \Delta_n^2 R(\Delta_n, \xtl,
\theta) \,,
\end{split}\\[0.5em]
\begin{split}\nonumber
& \hspace{-5mm} \EE^n_{i-1}\left( \partial_{\theta} g(\Delta_n,\xtr, \xtl, \theta)\right) \\
&=\Delta_n \left( \ll(\partial_{\theta}
g(0,\theta))(\xtl,\xtl) - \partial_{\theta} \ll_\theta(g(0,
\theta))(\xtl,\xtl) \right) \\
&\hspace{5mm} + \Delta_n^2 R(\Delta_n,\xtl, \theta)\,,
\end{split}\\[0.5em]
\begin{split}
&\hspace{-5mm} \EE^n_{i-1}\left( gg^\star (\Delta_n,\xtr,\xtl, \theta)
\right) 
\nonumber\\
&= \Delta_n \ll(gg^\star (0,\theta))(\xtl,\xtl)
+ \Delta_n^2 R(\Delta_n,\xtl, \theta)\,,
\end{split}
\end{align}
and, for $j,k,l,m=1,\ldots,d$,
\begin{align}
\EE^n_{i-1}\left( \left( \partial_{\theta} g\right)^2(\Delta_n,\xtr, \xtl, \theta)
\right) &= \Delta_n R(\Delta_n, \xtl, \theta) 
\nonumber\\
\EE^n_{i-1}\left( g_jg_kg_l(\Delta_n, \xtr, \xtl, \theta)
\right) &=
\Delta_nR(\Delta_n, \xtl, \theta)\nonumber\\
\EE^n_{i-1}\left( g_jg_kg_lg_m (\Delta_n,\xtr,\xtl,
    \theta) \right) &= \Delta_n R(\Delta_n, \xtl,
\theta)\,.\label{Eg4_jump}
\end{align}
\dqed
\label{lemma:condexpec}
\end{lemma}
\begin{proof}[\textbf{Proof of Lemma~\ref{lemma:condexpec}}]
Using (\ref{g2}),
(\ref{eqn:expectation_remainder_jump}), Remark
\ref{re:jump_gen_presentation}, and Lemmas~\ref{lemma:expansion_jump} and \ref{lemma:conseq} coordinate-wise, write
\begin{align*}
\begin{split}
& \hspace{-5mm} \EE^n_{i-1}\left( g(\Delta_n, \xtr,\xtl, \theta) \right) \\
& = \EE^n_{i-1}\left( g(0,\xtr,\xtl,\theta)\right) +
\Delta_n \EE^n_{i-1}( g^{(1)}(\xtr, \xtl,
      \theta)) \\
& \hspace{5mm} + \Delta_n^2 \EE^n_{i-1}\left(R(\Delta_n,\xtr, \xtl,
      \theta) \right) \\
& =g(0,\xtl,\xtl, \theta) + \Delta_n
\ll( g(0, \theta) )(\xtl,\xtl)
+ \Delta_n^2 R(\Delta_n, \xtl, \theta)
\\
 &\hspace{5mm} + \Delta_n\left( g^{(1)} (\xtl, \xtl,
      \theta) + \Delta_n R(\Delta_n, \xtl, \theta)
  \right) + \Delta_n^2 R(\Delta_n, \xtl,
      \theta) \\
&= \Delta_n\left( \ll( g(0, \theta) )(\xtl,
  \xtl) - \ll_{\theta}( g(0, \theta) )(\xtl,
  \xtl)\right) \\
&\hspace{5mm}+ \Delta_n^2 R(\Delta_n, \xtl,
\theta) 
\,,
\end{split}\\[0.5\baselineskip]
\begin{split}
& \hspace{-5mm} \EE^n_{i-1}\left( \partial_{\theta} g(\Delta_n,\xtr, \xtl, \theta)
  \right) \\
&= \EE^n_{i-1}\left( \partial_{\theta}
g(0,\xtr,\xtl, \theta) \right) +
\Delta_n \EE^n_{i-1} ( \partial_{\theta}g^{(1)}(\xtr,\xtl,
  \theta) ) \\
&\hspace{5mm} +
\Delta_n^2 \EE^n_{i-1}\left( R(\Delta_n, \xtr, \xtl, \theta) \right)\\
&=\partial_{\theta}
g(0,\xtl,\xtl, \theta) + \Delta_n \ll(\partial_{\theta}
g(0,\theta))(\xtl,\xtl) + \Delta_n^2 R(\Delta_n, \xtl, \theta) \\
&\hspace{5mm} + \Delta_n\left( \partial_{\theta}g^{(1)}(\xtl,\xtl,
  \theta) + \Delta_nR(\Delta_n,\xtl, \theta)\right) \\
&=\Delta_n \left( \ll(\partial_{\theta}
g(0,\theta))(\xtl,\xtl) - \partial_{\theta} \ll_\theta(g(0,
\theta))(\xtl,\xtl)\right)\\
&\hspace{5mm} + \Delta_n^2 R(\Delta_n,\xtl, \theta)\,,
\end{split}\\[0.5\baselineskip]
\begin{split}
&\hspace{-5mm} \EE^n_{i-1}\left( gg^\star (\Delta_n,\xtr,\xtl, \theta) \right)\\
&= \EE^n_{i-1}\left( gg^\star (0,\xtr,\xtl, \theta)
  \right)  \\
&\hspace{5mm} + \Delta_n \EE^n_{i-1}\left( 
  g^{(1)}g^\star (0,\xtr,\xtl, \theta) +
  g(g^{(1)})^\star (\xtr,\xtl, \theta)  \right) \\
&\hspace{5mm} + \Delta^2_n \EE^n_{i-1}\left( R(\Delta_n,\xtr,\xtl,
  \theta) \right) \\
&= gg^\star (0,\xtl,\xtl, \theta) +
\Delta_n \ll(gg^\star (0,\theta))(\xtl,\xtl) \\
&\hspace{5mm} + \Delta_n\left( 
  g^{(1)} g^\star (0,\xtl,\xtl, \theta) +
  g(g^{(1)})^\star (\xtl,\xtl, \theta)
\right) \\
&\hspace{5mm} + \Delta_n^2 R(\Delta_n,\xtl, \theta) \\
&= \Delta_n \ll(gg^\star (0,\theta))(\xtl,\xtl)
+ \Delta_n^2 R(\Delta_n,\xtl, \theta) \,.
\end{split}
\end{align*}
The three remaining equalities follow in a similar, more simple manner.
\end{proof}
Lemmas~\ref{lemma:generators_rate} and \ref{lemma:generator3_rate} provide expressions involving the infinitesimal
generator (\ref{infinitesimalgenerator2}). These lemmas may be used to prove
the subsequent lemmas on the expansion of conditional moments. The generalised Leibnitz formula for the $m$th derivative
of a product is useful for verifying these formulae. For proofs, see \cite[Section 3.A.2]{phdthesis}.
\begin{lemma}
Consider the model given by (\ref{eqn:SDEoptimal2}), with
  $A\subseteq \RR^2$ and $B\subseteq \RR$. Suppose that Assumptions~\ref{assumptions_on_X_jump} and
\ref{assumptions_on_g_jump}, and Condition~\ref{assumption_final} hold. 
Then, for $j=1,2,3$, the following holds for all $x\in \xx$
and $\theta \in \Theta$.
\begin{align*}
\ll \left(g_{j}g_3(0, \theta)\right)(x,x) &= 0
\end{align*}
and, furthermore,
\begin{align*}
\begin{split}
&\hspace{-5mm} \ll^2\left( g_{j}g_{3}(0,\theta)\right)(x,x) \\
&= \tfrac{3}{2}b^2(x, \betan) \left( 2a(x, \alphan) + \partial_y b^2(x,
  \betan)\right)
  \partial_y g_{j} \partial_y^2 g_{3}(0,x,x, \theta)\\
&\hspace{5mm} +\tfrac{1}{2} b^4(x, \betan)\left(2\partial_y g_{j} \partial_y^3 g_{3}  + 3\partial_y^2
  g_{j} \partial_y^2 g_{3} 
\right)(0,x,x, \theta)  \\
&\hspace{5mm} + \int_\RR \tfrac{1}{2} \left( b^2(
x+c(x,z,\alphan),
\betan)+ b^2(x, \betan) \left( 1
  + \partial_y c(x,z, \alphan)\right)^2 \right) \\
&\hspace{15mm} \times g_{j}\partial^2_y g_{3}(0,x+c(x,z,\alphan),x,
\theta)\, \nu_\alphan(dz) \,,
\end{split}\\[0.5em]
\begin{split}
& \hspace{-5mm} g_{j}^{(1)}(x,x, \theta)\\
&= -a(x, \alpha)\partial_y g_{j}(0,x,x, \theta) - \tfrac{1}{2}b^2(x,
  \beta)\partial^2_y g_{j}(0,x,x, \theta) \\
&\hspace{5mm} - \int_\RR g_{j}(0,x+c(x,z,\alpha),x, \theta)\,
  \nu_\alpha(dz)\,,
\end{split}\\[0.5em]
  \begin{split}
&\hspace{-5mm} \ll (
  g_{j}(0,\theta)g_{3}^{(1)}(\theta) )(x,x) \\
&=-\tfrac{1}{2} a(x, \alphan) b^2(x,
  \beta)\partial_y g_{j}\partial^2_y
g_{3}(0,x,x, \theta)\\
&\hspace{5mm}-\tfrac{1}{4}b^2(x,
  \beta) b^2(x,
\betan) \partial^2_y g_{j} \partial^2_y g_{3}(0,x,x, \theta)\\
&\hspace{5mm} +b^2(x,
\betan) \partial_y g_{j}(0,x,x, \theta)\partial_y
  g_{3}^{(1)}(x,x, \theta)\\
&\hspace{5mm}+ \int_\RR g_{j}(0,x+c(x,z,\alphan),x, \theta)
  g_{3}^{(1)}(x+c(x,z,\alphan), x, \theta) \,\nu_\alphan(dz)\,,
\end{split}\\[0.5em]
  \begin{split}
&\hspace{-5mm} \ll ( g_{j}^{(1)}(\theta)g_{3}(0,\theta) )(x,x) \\
&= -\tfrac{1}{2} a(x, \alpha)b^2(x,
\betan) \partial_y g_{j}\partial^2_y g_{3}(0,x,x, \theta) \\
&\hspace{5mm}-\tfrac{1}{4} b^2(x, \beta) b^2(x,
\betan) \partial^2_yg_{j}\partial^2_y g_{3}(0,x,x, \theta) \\
&\hspace{5mm} -\tfrac{1}{2} b^2(x,
\betan) \left(\int_\RR g_{j}(0,x+c(x,z,\alpha),x, \theta)\, \nu_\alpha(dz)\right) \partial^2_y g_{3}(0,x,x,
\theta)\,.
\end{split}
  \end{align*}
\dqed
\label{lemma:generators_rate}
\end{lemma}
\begin{lemma} 
Consider the model given by (\ref{eqn:SDEoptimal2}), with
  $A\subseteq \RR^2$ and $B\subseteq \RR$. Suppose that Assumptions~\ref{assumptions_on_X_jump} and
\ref{assumptions_on_g_jump}, and
Condition~\ref{assumption_final} hold. Then, for $j,k,l
= 1,2,3$ and $m=1,2$, the following nine expressions are equal to $0$ for all $x\in
\xx$ and $\theta \in \Theta$: 
\begin{displaymath}
\begin{array}{lll}
\ll ( g_{j}g_{k}g_{l}g_{3}(0,
  \theta))(x,x) & \hspace{7mm} \ll( g_{j}g_{k}g_{3}(0,
  \theta))(x,x) \\\ll(
  g_{j}^{(1)}(\theta) g_{k}g_{3}^2(0,\theta))(x,x) & \hspace{7mm} \ll ( g_{j}g_{k}g_{3}(0,\theta)g_{3}^{(1)}(\theta))(x,x)  \\ 
\ll^2(g_{j}g_3^3(0,\theta))(x,x)  & \hspace{7mm}
\ll( \partial_\alpha g_\beta(0, \theta))(x,x) \\
\ll( \partial_{\alpha_m} g_\beta(0, \theta)\partial_{\alpha_m}
  g_\beta^{(1)}(\theta))(x,x)  &  \hspace{7mm} 
\partial_\alpha \ll_\theta( g_\beta(0,\theta))(x,x)  \\ 
\ll^m( (\partial_\alpha g_\beta)^2(0, \theta))(x,x) \,.  & 
\end{array}
\end{displaymath}
\dqed
\label{lemma:generator3_rate}
\end{lemma}
\begin{lemma}
Consider the model given by (\ref{eqn:SDEoptimal2}), with
  $A\subseteq \RR^2$ and $B\subseteq \RR$. Suppose that Assumptions~\ref{assumptions_on_X_jump} and
\ref{assumptions_on_g_jump}, and
Condition~\ref{assumption_final} hold. 
Then, for $j =
1,2,3$, the following holds for all $n\in \NN$, $i=1,\ldots,n$, and $\theta \in \Theta$.
\begin{align*}
\begin{split}
&\hspace{-5mm} \EE^n_{i-1}\left( g_{j}g_{3}(\Delta_n,\xtr,\xtl,
  \theta)\right)\\
&=\Delta_n^2\left( \tfrac{1}{2} \ll^2\left( g_{j}g_{3}\left(0,
      \theta\right)\right)(\xtl,\xtl) +
  g_{j}^{(1)}g_{3}^{(1)}(\xtl,\xtl, \theta)\right. \\
&\hspace{12mm} + \left.  \ll( g_{j}(0,
\theta)g_{3}^{(1)}(\theta))(\xtl,\xtl) + \ll( g_{j}^{(1)}(
\theta)g_{3}(0, \theta))(\xtl,\xtl) \right)\\
&\hspace{5mm} + \Delta_n^3 R(\Delta_n,\xtl, \theta)\,.
\end{split}
\end{align*}
\dqed
\label{lemma:expectations_rate}
\end{lemma}
\begin{proof}[\textbf{Proof of Lemma~\ref{lemma:expectations_rate}}]
Lemmas~\ref{lemma:expansion_jump}, \ref{lemma:conseq}, and
\ref{lemma:generators_rate} are used to obtain
\begin{align*}
\begin{split}
&\hspace{-5mm} \EE^n_{i-1}\left( g_{j}g_{3}(\Delta_n,\xtr,\xtl,
  \theta)\right)\\
&= \EE^n_{i-1}\left( g_{j}g_{3}(0,\xtr,\xtl, \theta) \right) + \Delta_n  \EE^n_{i-1}( g_{j}(0,\xtr,\xtl,
\theta)g_{3}^{(1)}(\xtr,\xtl, \theta) ) \\
&\hspace{5mm} + \Delta_n\EE^n_{i-1}( g_{j}^{(1)}(\xtr,\xtl,
\theta)g_{3}(0,\xtr,\xtl, \theta)) \\
&\hspace{5mm} + \tfrac{1}{2}\Delta_n^2 \EE^n_{i-1}(
  g_{j}(0,\xtr,\xtl, \theta)g_{3}^{(2)}(\xtr,\xtl, \theta)) \\
&\hspace{5mm} + \Delta_n^2 \EE^n_{i-1}(
  g_{j}^{(1)}g_{3}^{(1)}(\xtr,\xtl, \theta) ) \\
&\hspace{5mm} + \tfrac{1}{2}\Delta_n^2 \EE^n_{i-1}(
  g_{j}^{(2)}(\xtr,\xtl, \theta)g_{3}(0,\xtr,\xtl, \theta) ) \\
&\hspace{5mm} + \Delta_n^3 \EE^n_{i-1}\left( R(\Delta_n,\xtr,\xtl, \theta) \right) \end{split}\\
\begin{split}
&= g_{j}g_{3}(0,\xtl,\xtl,
  \theta) + \Delta_n \ll( g_{j}g_{3}\left(0,
      \theta\right))(\xtl,\xtl)\\
&\hspace{5mm} + \tfrac{1}{2}\Delta_n^2 \ll^2\left( g_{j}g_{3}\left(0,
      \theta\right)\right)(\xtl,\xtl) \\
&\hspace{5mm} +  \Delta_ng_{j}(0,\xtl,\xtl,
\theta)g_{3}^{(1)}(\xtl,\xtl, \theta) \\
&\hspace{5mm} + \Delta_n^2 \ll( g_{j}(0,
\theta)g_{3}^{(1)}(\theta))(\xtl,\xtl) \\
&\hspace{5mm} +\Delta_n g_{j}^{(1)}(\xtl,\xtl,
\theta)g_{3}(0,\xtl,\xtl, \theta) \\
&\hspace{5mm} + \Delta_n^2 \ll( g_{j}^{(1)}(
\theta)g_{3}(0, \theta))(\xtl,\xtl) \\
&\hspace{5mm} +\tfrac{1}{2}\Delta_n^2 
  g_{j}(0,\xtl,\xtl, \theta)g_{3}^{(2)}(\xtl,\xtl, \theta) + \Delta_n^2 
  g_{j}^{(1)}g_{3}^{(1)}(\xtl,\xtl, \theta) \\
&\hspace{5mm} + \tfrac{1}{2}\Delta_n^2 
 g_{j}^{(2)}(\xtl,\xtl, \theta)g_{3}(0,\xtl,\xtl, \theta) + \Delta_n^3 R(\Delta_n,\xtl, \theta)
\end{split}\\
\begin{split}
&= \tfrac{1}{2}\Delta_n^2 \ll^2\left( g_{j}g_{3}\left(0,
      \theta\right)\right)(\xtl,\xtl) + \Delta_n^2 \ll( g_{j}(0,
\theta)g_{3}^{(1)}(\theta))(\xtl,\xtl) \\
&\hspace{5mm} + \Delta_n^2 \ll( g_{j}^{(1)}(
\theta)g_{3}(0, \theta))(\xtl,\xtl)+ \Delta_n^2 
  g_{j}^{(1)}g_{3}^{(1)}(\xtl,\xtl, \theta) \\
&\hspace{5mm} + \Delta_n^3 R(\Delta_n,\xtl, \theta)\,.
\end{split}
\end{align*}
\end{proof}
The proof of Lemma~\ref{lemma:expectations2_rate} utilises Lemma
\ref{lemma:generator3_rate}, and is otherwise
similar to the proof of Lemma~\ref{lemma:expectations_rate}.
For details, see \cite[Section 3.A.4]{phdthesis}.
\begin{lemma}
Consider the model given by (\ref{eqn:SDEoptimal2}), with
  $A\subseteq \RR^2$ and $B\subseteq \RR$. Suppose that Assumptions~\ref{assumptions_on_X_jump} and
\ref{assumptions_on_g_jump}, and Condition
\ref{assumption_final} hold. Then, for $j,k,l =
1,2,3$, the following holds for all $n\in \NN$, $i=1,\ldots,n$, and $\theta \in \Theta$.
\begin{align*}
\EE^n_{i-1}\left( g_{j}g_{k}g_{3}(\Delta_n,\xtr,\xtl,
  \theta)\right)
&= \Delta_n^2 R(\Delta_n,\xtl,
  \theta)\\
\EE^n_{i-1}\left( g_{j}g_{k}g_{l}g_{3}(\Delta_n,\xtr,\xtl,
  \theta)\right)
&= \Delta_n^2 R(\Delta_n,\xtl,
  \theta)\\
\EE^n_{i-1}\left( g_{j}g_{3}^3(\Delta_n,\xtr,\xtl,
  \theta)\right)
&= \Delta_n^3 R(\Delta_n,\xtl,
  \theta)\\
\EE^n_{i-1}\left( \partial_\alpha g_\beta(\Delta_n,\xtr,\xtl,
  \theta)\right) &= \Delta_n^2 R(\Delta_n,\xtl, \theta)\\
\EE^n_{i-1}\left( \left(\partial_\alpha g_\beta(\Delta_n,\xtr,\xtl,
  \theta)\right)^2 \right)
&= \Delta_n^3 R(\Delta_n,\xtl, \theta)
\end{align*}
\dqed
\label{lemma:expectations2_rate}
\end{lemma}

\section*{Acknowledgements}
We are grateful to the anonymous reviewers for their helpful and inspiring comments and suggestions which have improved this paper. Nina Munkholt Jakobsen was supported by the Danish Council for Independent Research -- Natural Science, through a grant to Susanne Ditlevsen. Michael S\o rensen was supported by the Center for Research in Econometric Analysis of Time Series funded by the Danish National Research Foundation. The research is part of the Dynamical Systems Interdisciplinary Network funded by the University of Copenhagen Programme of Excellence.
\bibliographystyle{apalike}
\bibliography{jumpdiff-1_v2} 
\end{document}